\documentclass[11pt,a4paper]{article}
\usepackage{a4wide}

\newcommand*{\Dataset}[1]{\textsc{#1}}
\usepackage[group-separator={,}]{siunitx}
\usepackage{booktabs}

\newcommand{\J}[1]{#1}

\usepackage{amsmath}
\usepackage{graphicx}
\usepackage{xcolor}
\usepackage{booktabs}
\usepackage{microtype}
\usepackage[natbib,color,cref,pdf,theorems]{nikis}
\usepackage{algorithm}
\usepackage{algpseudocodex}

\definecolor{solarizedYellow}{HTML}{B58900}
\definecolor{solarizedOrange}{HTML}{CB4B16}
\definecolor{solarizedRed}{HTML}{DC322F}
\definecolor{solarizedMagenta}{HTML}{D33682}
\definecolor{solarizedViolet}{HTML}{6C71C4}
\definecolor{solarizedBlue}{HTML}{268BD2}
\definecolor{solarizedCyan}{HTML}{2AA198}
\definecolor{solarizedGreen}{HTML}{859900}

\definecolor{teigiColor}{HTML}{5700B5}
\newcommand*{\teigi}[1]{{\color{teigiColor}\emph{#1}}}

\newcommand*{\PatternLegendS}[1]{\tikz[baseline=.35ex]{\draw[draw,#1] (0, 0) rectangle (.35, .35); }}

\usepackage{tikz}
\usetikzlibrary{patterns}

\newcommand*{\myblock}[1]{\vspace{0.0em}\noindent\textbf{#1.}}

\newcommand*{\intervalI}{\ensuremath{\mathcal{I}}} \newcommand*{\intervalJ}{\ensuremath{\mathcal{J}}} 
\newcommand*{\ibeg}[1]{\ensuremath{\textsf{b}(#1)}}
\newcommand*{\iend}[1]{\ensuremath{\textsf{e}(#1)}}

\newcommand*{\timeDICT}{\ensuremath{t_{\mathcal{D}}}}
\newcommand*{\timeSA}{\ensuremath{t_{\textrm{SA}}}}

\newcommand*{\VSub}{\ensuremath{v_{\intervalI}}}

\newcommand*{\LZDSub}{\ensuremath{z_{\mathup{D}\left[\intervalI\right]}}}
\newcommand*{\LZMWSub}{\ensuremath{z_{\mathup{MW}\left[\intervalI\right]}}}
\newcommand*{\EightSub}[1]{\ensuremath{z_{\mathup{78}\left[#1\right]}}}
\newcommand*{\RelativeSub}[1]{\ensuremath{z_{\mathup{rel}\left[#1\right]}}}
\newcommand*{\SevenSub}[1]{\ensuremath{z_{\mathup{77}\left[#1\right]}}}
\newcommand*{\runSub}[1]{\ensuremath{r_{\mathup{BWT}\left[#1\right]}}}

\newcommand*{\ac}[1]{#1}

\newcommand*{\block}[1]{\paragraph{#1}}

\newcommand*{\dst}{\mathsf{dst}}
\newcommand*{\src}{\mathsf{src}}
\newcommand*{\iLexparse}{\strategyname{lex-parse}}

\newcommand*{\iFPA}{\strategyname{FPA}}
\newcommand*{\iLZWFP}{\strategyname{LZW-FP}}
\newcommand*{\iLZEightFP}{\strategyname{FP78}}
\newcommand*{\iLZEightFPA}{\strategyname{FPA78}}
\newcommand{\RmQ}{\functionname{RmQ}}
\newcommand{\RMQ}{\functionname{RMQ}}

\newcommand*{\strategyname}[1]{\textsf{#1}} \newcommand*{\instancename}[1]{\ensuremath{\mathsf{#1}}} \newcommand*{\functionname}[1]{{{\renewcommand{\rmdefault}{ptm}\fontfamily{ppl}\selectfont\textrm{\textup{#1}}}}} 

\newcommand{\select}[1][]{\operatorname{select}_{#1}}
\newcommand{\rank}[1][]{\operatorname{rank}_{#1}}

\newcommand*{\iTernary}{\strategyname{ternary}}
\newcommand*{\iCics}{\strategyname{cics}}

\newcommand*{\lce}{\functionname{lce}}
\newcommand*{\fnChild}{\functionname{child}}
\newcommand*{\fnSuffixLink}{\functionname{suffixlink}}
\newcommand*{\parent}{\functionname{parent}}
\newcommand*{\fnDepth}{\functionname{depth}}
\newcommand*{\fnSelectLeaf}{\functionname{select\_leaf}}
\newcommand*{\rangeNextValue}{\functionname{range\_next\_value}}
\newcommand*{\strdepth}{\functionname{str\_depth}}
\newcommand*{\lmostleaf}{\ensuremath{\functionname{range}_\textup{L}}}
\newcommand*{\rmostleaf}{\ensuremath{\functionname{range}_\textup{R}}}
\newcommand{\levelanc}   {\functionname{level\_anc}}

\newcommand*{\LPF} {\instancename{LPF}}
\newcommand*{\LCP} {\instancename{LCP}}
\newcommand*{\PLCP}{\instancename{PLCP}}
\newcommand*{\ISA} {\instancename{ISA}}

\newcommand*{\ST}  {\instancename{ST}}
\newcommand*{\SA}  {\instancename{SA}}

\newcommand*{\DICT}  {\ensuremath{\mathcal{D}}}

\newcommand*{\runningExample}{ababbababbabb}

\newcommand*{\IC}{..}\newcommand*{\ICS}{..}

\author{Dominik K\"{o}ppl}
\date{University of Yamanashi, K\={o}fu, Japan}
\title{Substring Compression Variations and LZ78-Derivates\thanks{Parts of this work have already been presented at
    the Data Compression Conference 2024~\cite{koppl24computing}.
}}

\begin{document}
\maketitle
\begin{abstract}
  We propose algorithms computing the semi-greedy Lempel--Ziv 78 (LZ78), 
  the Lempel--Ziv Double (LZD), and the Lempel--Ziv--Miller--Wegman (LZMW)
  factorizations in linear time for integer alphabets. 
  For LZD and LZMW, we additionally propose data structures that can be constructed in linear time, which can solve the substring compression problems for these factorizations in time linear in the output size.
  For substring compression, we give results for lexparse and closed factorizations.
\end{abstract}
{\small
\textbf{Keywords:} lossless data compression, factorization algorithms, substring compression
}

\section{Introduction}

The substring compression problem~\cite{cormode05substringcompression}
is to preprocess a given input text $T$ such that computing a compressed version of a substring of $T[i..j]$ 
can be done efficiently.
This problem has originally been stated for the Lempel--Ziv-77 (LZ77) factorization\J{~\cite{storer78macro}},
but extensions to the generalized LZ77 factorization{~\cite{keller14generalized}}, 
the Lempel--Ziv 78 factorization~\cite{koppl21nonoverlapping},
the run-length encoded Burrows--Wheeler transform (RLBWT)~\cite{babenko15wavelet},
and the relative LZ factorization{~\cite[Sect.~7.3]{dissKociumakaTomasz}}
have been studied.
Given $n$ is the length of $T$,
a trivial solution is to precompute the compressed output of $T[\intervalI]$ for all intervals $\intervalI \subset [1..n]$.
This however gives us already $\Om{n^2}$ solutions to compute and store.
We therefore strive to find data structures within $\oh{n^2}$ words of space, more precisely: \Oh{n \lg n} bits of space, 
that can answer a query in time linear in the output size with a polylogarithmic term on the text length. 
We investigate variations of the LZ78 factorization, namely
LZD~\cite{goto15lzd} and LZMW~\cite{miller85variations}, which have not yet been studied with regard to that aspect.
\subsection{Related Work}
In what follows, we briefly highlight work in the field of substring compression, and then list work related to the LZ78 derivations we study in this paper.

\begin{table}[t]
  \centering
  \caption{Results on various types of substring compression problems.
    For a given query interval $\intervalI$, $z$ denotes the output size (e.g., the number of factors) and $s = |\intervalI|$ the length of the query interval.
    For generalized Lempel--Ziv (gLZ) we have a second query interval $\intervalJ$ with length $s'$.
    $e$ denotes the number of edges of the compressed acyclic word graph (CDAWG)~\cite{blumer85dawg}.
    Space is given in words.
  }
  \label{tabSubstringCompression}
  \begin{tabular}{lllll}
\toprule
    problem & space & constr. time & query time & ref.
\\\midrule
    LZ77 & {\Oh{n \lg^\epsilon n}} & {\Oh{n \lg n}} & {\Oh{z \lg n \lg \lg n}} & \cite{cormode05substringcompression} \\
	 & {\Oh{n \lg^\epsilon n}} & - & {\Oh{z \lg\lg n}} & \cite[Thm.~2]{keller14generalized}\\
	 & {\Oh{n}} &  - &  {\Oh{z \lg^\epsilon n}} & \cite[Thm.~2]{keller14generalized}\\
    gLZ & {\Oh{n \log^\epsilon n}} & - & {\Oh{z \lg (s'/z) \lg\lg n}} & \cite[Thm.~4]{keller14generalized}\\
       & {\Oh{n}} & - & {\Oh{z \lg (s'/z) \lg^\epsilon n}} & \cite[Thm.~4]{keller14generalized}\\
       & {\Oh{n}} & - & {\Oh{z \lg^\epsilon n}} & \cite[Lemma~4]{abedin20rangelcp}\\
    RLBWT & 
    {\Oh{n}} &
    {{\Oh{n \sqrt{\lg n}}}} expected &
      {\Oh{\runSub{\intervalI} \lg |\intervalI|}} &
       \cite{babenko15wavelet} \\
      Lyndon &
    {\Oh{n}} &
    {\Oh{n}} &
    {\Oh{z}}&
      \cite{kociumaka16minsuf}\\
    LZ78&
    {\Oh{n}} &
    {\Oh{n}} &
    {\Oh{z}}&
      \cite{koppl21nonoverlapping}\\
    &
    {\Oh{e}} &
    {\Oh{n\lg n}} &
    {\Oh{z}}&
      \cite{shibata24lz78}\\
    LZD&
    {\Oh{n}} &
    {\Oh{n}} &
    {\Oh{z}}&
    \cref{thmLZD}\\
    LZMW&
    {\Oh{n}} &
    {\Oh{n}} &
    {\Oh{z}}&
    \cref{thmLZMW}\\
    lexparse&
    {\Oh{n}} &
    {\Oh{n \sqrt{\lg n}}} &
    {\Oh{z \lg s}}&
    \cref{thmLexparse}\\
    closed fact.\ &
    {\Oh{n}} &
    {\Oh{n}} &
    {\Oh{z \lg n}}&
    \cref{thmClosedFact}\\
			\bottomrule
\end{tabular}
\end{table}

\myblock{Substring Compression}
\citet{cormode05substringcompression} solved the substring compression problem for \ac{LZ77} with a data structure 
answering the query for $\intervalI$ in \Oh{\SevenSub{\intervalI} \lg n \lg \lg n} time,
where \SevenSub{\intervalI} denotes the number of produced \ac{LZ77} factors of the queried substring~$T[\intervalI]$.
Their data structure uses \Oh{n \lg^\epsilon n} space, and can be constructed in \Oh{n \lg n} time.
This result has been improved by \citet{keller14generalized} to \Oh{\SevenSub{\intervalI} \lg\lg n} query time for the same space, 
or to \Oh{\SevenSub{\intervalI} \lg^\epsilon n} query time for linear space.
They also gave other trade-offs regarding query times and the sizes of the used data structure.
\citeauthor{keller14generalized} also introduced the generalized substring compression query problem,
a variant of the relative LZ compression~\cite{kuruppu10relative}. 
Given two intervals $\intervalI$ and $\intervalJ$ as query input, the task is to compute the Lempel--Ziv parsing of $T[\intervalI]$ relative to $T[\intervalJ]$ (meaning that we compress $T[\intervalI]$ with references only based on substrings in $T[\intervalJ]$).
Their results are similar to the \ac{LZ77} case:
\Oh{\RelativeSub{\intervalI,\intervalJ} \lg (|\intervalJ|/\RelativeSub{\intervalI,\intervalJ}) \lg\lg n} query time for $\Oh{n \log^\epsilon n}$ space, and 
\Oh{\RelativeSub{\intervalI,\intervalJ} \lg (|\intervalJ|/\RelativeSub{\intervalI,\intervalJ}) \lg^\epsilon n} query time for $\Oh{n}$ space,
where $\RelativeSub{\intervalI,\intervalJ}$ denotes the number of factors in the Lempel--Ziv factorization of $T[\intervalI]$ relative to $T[\intervalJ]$.
The generalized LZ77 compression method is also known as relative LZ~\cite{kuruppu10relative}.

The main idea of tackling the problem for \ac{LZ77} 
is to use a data structure answering interval LCP queries, 
which are usually answered by two-dimensional range successor/predecessor data structures.
Much effort has been put in devising such data structures~\cite{patil13rangelcp, amir14rangelcp, amir15rangelcp, abedin20rangelcp}.
Most recently, \citet{matsuda20compressed} proposed a data structure answering an interval LCP query in \Oh{n^\epsilon} time 
while taking \Oh{\epsilon^{-1} n(H_0(T) +1)} bits of space, where $H_0$ denotes the zeroth order empirical entropy.
Therefore, they can implicitly answer an \ac{LZ77} substring compression query in \Oh{\SevenSub{\intervalI} n^\epsilon} time within compressed space.
Recently, \citet{bille20string} proposed data structures storing the \ac{LZ77}-compressed suffixes of~$T$ for answering a pattern matching query of an \ac{LZ77}-compressed pattern~$P$ without decompressing~$P$.
Their proposed data structures seem also to be capable to answer substring compression queries.

The substring compression problem has also been studied~\cite{babenko15wavelet} for another compression technique, 
the RLBWT\J{~\cite{burrows94bwt}}:
\citet{babenko15wavelet} showed how to compute the RLBWT of $T[\intervalI]$ in $\Oh{\runSub{\intervalI} \lg |\intervalI|}$ time,
where $\runSub{\intervalI}$ denotes the number of character runs in the BWT of $T[\intervalI]$.
\J{Another kind of factorization is the Lyndon factorization~\cite{chen58lyndon},
for which \citet{kociumaka16minsuf} gave an algorithm that can compute the Lyndon factorization of a substring in time linear in the number of factors.}
Recently, \citet{koppl21nonoverlapping} proposed data structures for the substring compression problem with respect to the LZ78 factorization.
These data structures can compute the \ac{LZ78} factorization of $T[\intervalI]$
\sitemize*{\item in \Oh{\EightSub{\intervalI}} time using \Oh{n \lg n} bits of space, or
		\item in \Oh{\EightSub{\intervalI} (\log_\sigma n + \lg \EightSub{\intervalI})} time using \Oh{n \lg \sigma} bits of space,
}
where \EightSub{\intervalI} is the number of computed \ac{LZ78} factors, and $\epsilon \in (0,1]$ a selectable constant.
	Finally, \citet{kociumaka23internal} studied tools for internal queries that can be constructed in $\Oh{n/\log_\sigma n}$ time optimally in the word RAM model, which also led to new complexities for the LZ77 substring compression.

Another related research field covers internal queries such as queries for the longest palindrome~\cite{amir20dynamic,mitani23internal} or longest common substring,
pattern matching~\cite{kociumaka15internal},
counting of distinct patterns~\cite{charalampopoulos20counting}, quasi-periodicity testing~\cite{crochemore20internal}, and range shortest unique substring queries~\cite{abedin19rangesus},
all inside a selected substring of the indexed text.

\myblock{LZ78 Derivates}
In this paper, we highlight three factorizations deriving from the LZ78 factorization:
(a) the flexible parsing variants~\cite{horspool95effect,matias99optimality,matias01effect} of LZ78,
(b) LZMW~\cite{miller85variations} and 
(c) LZD~\cite{goto15lzd}.
The first (a) achieves the fewest number of factors among all LZ78 parsings that have an additional choice instead of strictly following the greedy strategy to extend the longest possible factor by one additional character.
The last two factorizations (b) and (c) let factors refer to \emph{two} previous factors, 
and are noteworthy variations of the LZ78 factorization.
For instance, LZ78 factorizes $T = \texttt{a}^n$ into $\Ot{\sqrt{n}}$ factors, while both 
variations have $\Ot{\lg n}$ factors, where the factor lengths scale with a power of two or with the Fibonacci numbers for LZD and LZMW, respectively.
That is because LZD allows the selection of two references making it possible to form factors of the form $(\texttt{a}^k, \texttt{a}^k) = \texttt{a}^{k+1}$,
while LZMW requires such a selection to be for subsequent factors.
Here, the best is to factorize $\texttt{a}^n$ in lengths of the Fibonacci numbers 
$\texttt{a}^1, \texttt{a}^1, \texttt{a}^2, \texttt{a}^3, \texttt{a}^5, \texttt{a}^8, \texttt{a}^{13}, \ldots$ since then the length of the two preceding factors is maximized.
For the Fibonacci numbers it is known that they grow exponentially such that the number of LZMW factors of $\texttt{a}^n$ is $\Ot{\lg n}$.
The best lower bound on the number of LZMW factors with binary alphabets can be achieved with a cousin $G_k$ of the Fibonacci words, defined by $G_k = G_{k-2} G_{k-1}$, $G_1 = \texttt{a}$, and $G_2 = \texttt{b}$.
Then the concatenated text $G_1 G_2 \cdots G_k$ has $k$ LZMW factors.

On the downside, \citet[Thm~1.1]{de24grammar} have shown that, while random access (i.e., accessing a character of the original input string) is possible in \Oh{\lg \lg n} time on LZ78-compressed strings,
this is not possible on LZD-compressed strings without increasing the space significantly.
With respect to factorization algorithms,
\citet{goto15lzd} can compute LZD in \Oh{n \lg \sigma} time with \Oh{n \lg n} bits of space, or in \Oh{z \lg n} bits of space.
LZD and LZMW were studied by \citet{badkobeh17two}, 
who gave a bound of $\Om{n^{5/4}}$ time for the latter factorization algorithm~\cite{goto15lzd} using only $\Ot{z \lg n}$ bits of working space, where $z$ denotes the output size of the respective factorization.
Interestingly, the same lower bound holds for the original LZMW algorithm.
They also gave Las Vegas algorithms for computing the factorizations in $\Oh{n + z \lg^2 n}$ expected time and $\Oh{z \lg^2 n}$ space.

\myblock{Our Contribution}
In what follows, we propose construction algorithms for the three aforementioned types of factorizations.
For LZD and LZMW, we also study their substring compression problem.
Regarding these two factorizations, despite having an $\Om{n^{5/4}}$ time bound on the running time of the known deterministic algorithms, 
we leverage the data structure of~\cite{koppl21nonoverlapping} to answer a
substring compression query for each of the two factorizations in \Oh{z} time using \Oh{n \lg n} bits of space, where $z$ again denotes the number of factors of the corresponding factorization.
Since the used data structure can be computed in linear time, this also leads to the first deterministic linear-time computation of LZD and LZMW;
the aforementioned Las Vegas algorithm of \citet{badkobeh17two} is only linear \emph{in expectancy} for $z \in \oh{n/ \lg^2 n}$\footnote{Assuming $z \in \oh{n/ \lg^2 n}$ is reasonable for relatively compressible strings.}.
Additionally, we can compute the flexible parse of LZ78 in the same complexities, or 
in $\Oh{n (\lg z + \log_\sigma n)}$ time within $\Oh{n \lg \sigma}$ bits of space.
This holds if we work with the flexible parsing variant of \citet{horspool95effect} or \citet{matias99optimality}.
Best previous results have linear expected time or are only linear for constant alphabet sizes~\cite{matias01effect}.

\block{Structure of this Article}
After the preliminaries in \cref{secPrelimanaries}, we focus on the substring compression for lexparse (\cref{secLexparse}) and the closed factorizations (\cref{secClosedFact}),
which serve as an appetizer for what follows.
For the closed factorization, we also make our first acquaintance with a solution using a suffix tree.
Subsequently, in \cref{secSemiGreedy}, we introduce the flexible parsing in two different flavors for LZ78,
and show how to compute both parsings with the AC-automaton and suffix trees in \cref{secACautomaton} and \cref{secFlexParseLinTime}, respectively.
The final part of the theoretical analysis is \cref{lzEightVariations}, in which we propose suffix tree-based solutions for computing the LZD and LZMW parsings.
The remainder of this article covers two practical benchmarks in \cref{secEvaluation}.
First, we study the number of factors of the two flexible parsing variations compared to standard LZ78 in \cref{secEvalFlexParse}.
Second, we evaluate the computation time for LZ78 substring compression when using an already computed suffix tree versus a standard LZ78 compressor in \cref{secEvalEight}.
A conclusion with outlook for future work and open problems ends the article in \cref{secOutlook}.
Compared to the conference version~\cite{koppl24computing}, we added the closed factorizations, the approach with the AC automaton,
improved the structure by adding more examples and figures, and cover an evaluation in \cref{secEvaluation}.

\section{Preliminaries}\label{secPrelimanaries}
With $\lg$ we denote the logarithm~$\log_2$ to base two.
Our computational model is the word RAM model with machine word size $\Om{\lg n}$ bits for a given input size~$n$.
Accessing a word costs $\Oh{1}$ time.

Let $T$ be a text of length $n$ whose characters are drawn from an integer alphabet $\Sigma= [1\IC{}\sigma]$ 
with $\sigma \le n^{\Oh{1}}$.
Given $X,Y,Z \in \Sigma^*$ with $T = XYZ$, 
then $X$, $Y$ and $Z$ are called a \teigi{prefix}, \teigi{substring} and \teigi{suffix} of $T$, respectively.
A substring $Y$ of $T$ is called \teigi{proper} if $Y \neq T$.
We call $T[i\ICS{}]$ the $i$-th suffix of $T$, and denote a substring $T[i] \cdots T[j]$ with $T[i\IC{}j]$.
A \teigi{parsing dictionary} is a set of strings.
A parsing dictionary~$\mathcal{D}$ is called \teigi{prefix-closed} if it contains, for each string $S \in \mathcal{D}$,
all prefixes of $S$ as well.
A \teigi{factorization} of $T$ of size~$z$ partitions $T$ into $z$~substrings $F_1 \cdots F_z = T$.
Each such substring $F_x$ is called a \teigi{factor}.
Let $\dst_x := |F_1..F_{x-1}|+1$ denote the starting position of factor~$F_x$.

Stipulating that $F_0$ is the empty string,
a factorization $F_1\cdots F_z = T$ is called the \teigi{LZ78 factorization}\J{~\cite{ziv78lz}} of $T$ iff, for all $x \in [1..z]$,
the factor~$F_x$ is the longest prefix of $T[|F_1\cdots F_{x-1}|+1..]$ such that $F_x = F_y \cdot c$ for some $y \in [0..x-1]$ and $c \in \Sigma$,
that is, $F_x$ is the longest possible previous factor $F_y$ appended by the following character in the text.
The dictionary for computing $F_x$ is $\DICT = \{F_y \cdot c : y \in [0..x-1], c \in \Sigma \}$.
Formally,
$y = \argmax \{ \abs{F_{y'}} : F_{y'} = T[\dst_x..\dst_x+|F_{y'}|-1]\}$.
We say that $y$ and $F_y$ are the \teigi{referred index} and the \teigi{referred factor} of the factor~$F_x$, respectively.
We call a factor of length one \teigi{literal}; such a factor has always the referred index~$0$.
The other factors are called \teigi{referencing}.
See \cref{figFactLZ78} for an example of the LZ78 factorization.

\begin{figure}
  \begin{minipage}{0.45\linewidth}
  \includegraphics[page=1]{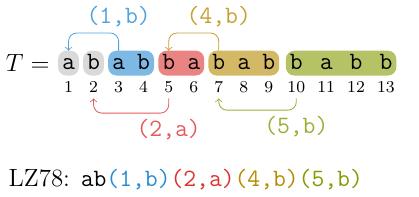}
  \end{minipage}
  \begin{minipage}{0.5\linewidth}
  \caption{LZ78 factorization of $T = \texttt{ababbababbabb}$,
    given by
    $F_1 =\texttt{a}$,
    $F_2 =\texttt{b}$,
    $F_3 =\texttt{ab}   = F_1 \cdot \texttt{b}$,
    $F_4 =\texttt{ba}   = F_2 \cdot \texttt{a}$,
    $F_5 =\texttt{bab}  = F_4 \cdot \texttt{b}$,  and
    $F_6 =\texttt{babb} = F_5 \cdot \texttt{b}$.
    The figure shows an encoding of the factors, where a number denotes the index of the referred factor.
  }
  \label{figFactLZ78}
  \end{minipage}
\end{figure}

\begin{figure}[htpb]
  \begin{minipage}{0.38\linewidth}
  \includegraphics{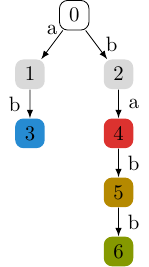}
  \end{minipage}
  \begin{minipage}{0.6\linewidth}
  \caption{The LZ trie of our running example $T = \texttt{\runningExample}$. 
    The coloring of the LZ trie matches the colors of the LZ78 factors in \cref{figFactLZ78} visualizing the LZ78 factorization of~$T$.
  }
  \label{figLZTrie}
  \end{minipage}
\end{figure}

All factors $F_x$ are distinct except maybe the last factor~$F_z$,
which needs to be treated as a border case. 
In what follows, we omit this border case analysis for the sake of simplicity (in LZ78 as well as in all later explained variants).
If $T$ terminates with a character appearing nowhere else in $T$, then the last factor is also distinct from the others.
The LZ trie represents each LZ factor as a node (the root represents the factor~$F_0$).
The node representing the factor $F_y$ has a child representing the factor $F_x$ connected with an edge labeled by a character~$c \in \Sigma$ if and only if $F_x = F_y c$.
See \cref{figLZTrie} for the LZ trie of our running example.

When computing $F_x$, the LZ78 parsing dictionary \DICT{} contains all previous factors (including $F_0$) with all possible character extensions such that determining $F_x$ can be reduced to finding the longest element $F_y \cdot c$ in \DICT{} that matches with $F_x$.

An \teigi{interval} $\intervalI=[b..e]$ is the set of consecutive integers from $b$ to $e$, for $b\le e$. 
For an interval $\intervalI$, we use the notations $\ibeg{\intervalI}$ and $\iend{\intervalI}$ to denote the beginning and the end of $\intervalI$,
i.e., $\intervalI = [\ibeg{\intervalI}..\iend{\intervalI}]$.
We write $\abs{\intervalI}$ to denote the length of $\intervalI$; i.e., $\abs{\intervalI}=\iend{\intervalI}-\ibeg{\intervalI}+1$,
and $T[\intervalI]$ for the substring $T[b..e]$.

\myblock{Text Data Structures}
Let $\SA$ denote the \teigi{suffix array}~\cite{manber93sa} of $T$.
The entry~$\SA[i]$ is the starting position of the $i$-th lexicographically smallest suffix such that $T[\SA[i] \IC] \prec T[\SA[i+1] \IC]$ for all integers~$i \in [1\IC{}n-1]$.
Let $\ISA$ of $T$ be the inverse of $\SA$, i.e., $\ISA[\SA[i]] = i$ for every $i \in [1\IC{}n]$.
The \teigi{LCP array} is an array with the property that $\LCP[i]$ is the length of the longest common prefix (LCP) of $T[\SA[i] \IC]$ and $T[\SA[i-1] \IC]$ for every $i \in [2\IC{}n]$.
For convenience, we stipulate that $\LCP[1] := 0$.
The array~$\Phi$ is defined as $\Phi[i] := \SA[\ISA[i]-1]$, and $\Phi[i] := n$ in case that $\ISA[i] = 1$.
The \teigi{permuted LCP} (\teigi{PLCP}) \teigi{array} $\PLCP$ stores the entries of $\LCP$ in text order, i.e., $\PLCP[\SA[i]] = \LCP[i]$.
\J{See \cref{tabDatastructures} for an example of some of the introduced data structures.}
Given an integer array $Z$ of length $n$ and an interval $[i..j] \subset [1..n]$,
the range minimum query $\RmQ_Z(i, j)$ (resp.\ the range maximum query $\RMQ_Z(i, j)$)
asks for the index~$p$ of a minimum element (resp.\ a maximum element) of
the subarray $Z[i..j]$, i.e.,
$p \in \arg\min_{i\leq k \leq j} Z[k]$, or respectively $p \in \arg\max_{i\leq k \leq j} Z[k]$.
We use the following data structure to handle this kind of query:

\begin{lemma}[\cite{davoodi12rmq}]\label{lemRMQ}
    Given an integer array $Z$ of length $n$,
    there is an $\RmQ$ (resp.\ $\RMQ$) data structure taking $2n+o(n)$ bits of space that can answer an $\RmQ$ (resp.\ $\RMQ$) query on $Z$ in constant time.
    This data structure can be constructed in $\Oh{n}$ time with $o(n)$ bits of additional working space.
\end{lemma}

An \teigi{LCE} (\teigi{longest common extension}) query $\lce(i,j)$ asks for the longest common prefix of two suffixes $T[i..]$ and $T[j..]$.
We can answer LCE queries in $\Oh{1}$-query time with an \RmQ{} data structure built on the LCP array 
because $\lce(i,j)$ is the LCP value of the \RmQ{} in the range $[\ISA[i]+1..\ISA[j]]$.

Given a character $\texttt{c} \in \Sigma$, and an integer~$j$, 
the \teigi{rank} query $T.\rank[\texttt{c}](j)$ counts the occurrences of \texttt{c} in $T[1\IC{}j]$, and 
the \teigi{select} query $T.\select[\texttt{c}](j)$ gives the position of the $j$-th \texttt{c} in $T$, if it exists.
We stipulate that $\rank[\texttt{c}](0) = \select[\texttt{c}](0) = 0$.
If the alphabet is binary, i.e., when $T$ is a bit vector,
there are data structures~\cite{jacobson89rank,clark96select} that use \oh{|T|} extra bits of space, and
can compute $\rank{}$ and $\select{}$ in constant time, respectively.
There are representations~\cite{baumann19rankselect} with the same constant-time bounds that can be constructed in time linear in $|T|$. 
We say that a bit vector has a \teigi{rank-support} and a \teigi{select-support} if it is endowed 
by data structures providing constant time access to $\rank$ and $\select$, respectively.

\myblock{Suffix Tree}
From now on, we assume that we have appended, to the input text $T$, a special character~$\texttt{\$}$ 
smaller than all other characters appearing in~$T$ that is not subject to a query.
By the property of $\texttt{\$}$, there is no suffix of~$T$ having another suffix of~$T$ as a prefix.
The \teigi{suffix trie} of $T$ is the trie of all suffixes of $T$. 
There is a one-to-one relationship between the suffix trie leaves and the suffixes of~$T$.
The \teigi{suffix tree}~\ST{} of $T$ is the tree obtained by compacting the suffix trie of $T$.
Like the suffix trie, the suffix tree has $n$ leaves, but the number of internal nodes of the suffix tree is at most $n$ because every \ST{} node is branching.
The string stored in a suffix tree edge~$e$ is called the \teigi{label} of $e$.
The \teigi{string label} of a node $v$ is defined as the concatenation of all edge labels on the path from the root to $v$;
its \teigi{string depth}, denoted by $\strdepth(v)$, is the length of its string label.
The leaf corresponding to the $i$-th suffix~$T[i\ICS{}]$ is labeled with the \teigi{suffix number}~$i \in [1\IC{}n]$.
The \teigi{leaf-rank} is the preorder rank ($\in [1\IC{}n]$) of a leaf among the set of all \ST{} leaves. For instance, the leftmost leaf in \ST{} has leaf-rank 1, while the rightmost leaf has leaf-rank~$n$.
Reading the suffix numbers stored in the leaves of \ST{} in leaf-rank order gives the suffix array.
\cref{figSuffixTree} depicts the suffix tree of our running example.

Suffix trees can be computed in linear time~\cite{farach-colton00st}, take $\Oh{n \lg n}$ bits of space,
and can be augmented to support the following operations in constant time (cf.~\cite{fischer18lz} for details):

\begin{itemize}
  \item $\fnDepth(v)$ returns the depth of an \ST{} node~$v$.
  \item $\lmostleaf(v)$ and $\rmostleaf(v)$ return the leaf-rank of the leftmost and the rightmost leaf of the subtree rooted at an \ST{} node~$v$, respectively.
  \item $\levelanc(\lambda,d)$ selects the ancestor of the \ST{} leaf~$\lambda$ at depth~$d$.
\end{itemize}

We further want to answer the following query.

\begin{itemize}
  \item $\strdepth(v)$ returns the string depth of an \emph{internal} node.
\end{itemize}
Within $\Oh{n \lg n}$ bits of space, we can answer $\strdepth(v)$ by taking the minimal \LCP{} value in the \SA{} range representing the leaves in $v$'s subtree, i.e.,
the value of $\LCP.\RmQ[\lmostleaf(v)+1..\rmostleaf(v)]$.
That is because, for each $i \in [\lmostleaf(v)..\rmostleaf(v)]$, the string label of $v$ is a prefix of the suffix $T[i..]$, but not prefix of any other suffix.
When storing $\LCP$ is too costly, we can get access to $\LCP$ with $\PLCP$ and the compressed suffix array~\cite{grossi05csa} within $\Oh{\log_\sigma n}$ access time.
Stipulating that \timeSA{} denotes the time to access \SA{}, we can answer $\strdepth(v)$ in \Oh{\timeSA} time.
Without $\LCP$ and the full suffix array, there is an $\Oh{n \lg \sigma}$-bits representation of the suffix tree~\cite{munro17cst,fischer18lz}, which can be computed in linear time.
In what follows, we analyze our algorithms using suffix trees under both settings, having $\Oh{n \lg n}$ bits with the above queries at constant time, 
or $\Oh{n \lg \sigma}$ bits with a slower $\Oh{\log_\sigma n}$ time for $\strdepth(v)$.
Similarly, we can perform the following operation in \Oh{\timeSA} time:
\begin{itemize}
  \item $\fnSelectLeaf(i)$ returns the \ST{} leaf with suffix number~$i$.
\end{itemize}

\begin{figure}[t]
  \begin{minipage}{0.45\linewidth}
  \includegraphics[width=\textwidth,page=3]{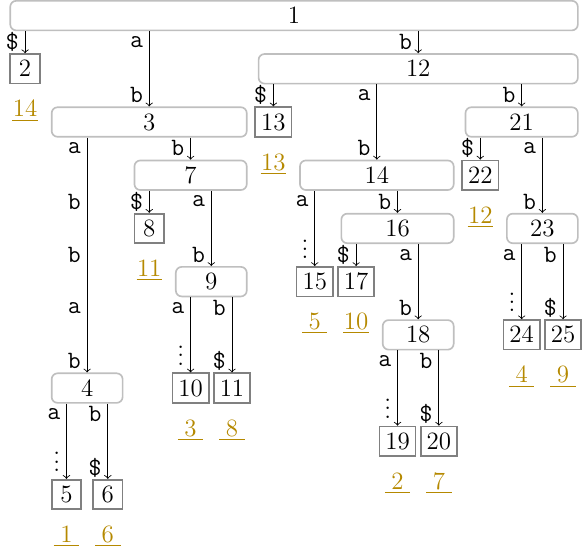}
  \end{minipage}
  \hfill
  \begin{minipage}{0.45\linewidth}
    \caption{The suffix tree of $T = \texttt{\runningExample}$.
    Edges to leaves with labels longer than two are trimmed with `$\vdots$` to save space on the paper.
    Nodes are labeled by their preorder ranks.
    Each leaf is decorated with its suffix number with a golden (\protect\PatternLegendS{fill=solarizedYellow}) underlined number beneath.
  }
  \label{figSuffixTree}
  \end{minipage}
\end{figure}

As a warm-up we start with the substring compression for lexparse and closed factorizations.

\begin{table}
  \centering
  \caption{Some introduced text data structures on the string $T = \texttt{ababbababbabb}$.}
  \label{tabDatastructures}
  \begin{tabular}{l*{14}{r}}
    \toprule
$i$&1&2&3&4&5&6&7&8&9&10&11&12&13\\
\midrule
$T$&a&b&a&b&b&a&b&a&b&b&a&b&b\\
$\ISA$&1&9&4&12&7&2&10&5&13&8&3&11&6\\
$\Phi$&-&10&11&12&13&1&2&3&4&5&6&7&8\\
$\PLCP$&0&4&3&2&1&7&6&5&4&3&2&1&0\\
\midrule
$\SA$&1&6&11&3&8&13&5&10&2&7&12&4&9\\
$\LCP$&0&7&2&3&5&0&1&3&4&6&1&2&4\\
\bottomrule
  \end{tabular}
\end{table}

\begin{figure}
  \begin{minipage}{0.48\linewidth}
  \includegraphics[page=4]{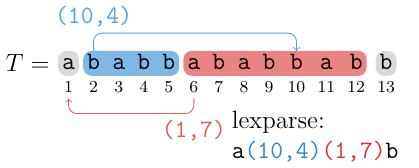}
  \end{minipage}
  \hfill
  \begin{minipage}{0.50\linewidth}
  \caption{lexparse of our running example $T = \texttt{ababbababbabb}$,
    given by
    $F_1 =\texttt{a}$,
    $F_2 =\texttt{babb} = T[10..13]$,
    $F_3 =\texttt{ababbab} = T[1..7]$, and
    $F_4 =\texttt{b}$.
    Non-literal factors are encoded by text position and length of the substring in $T$ they refer to.
  }
  \label{figFactLexparse}
  \end{minipage}
\end{figure}

\section{Warm-Ups: Lexparse and Closed Factorization}
The \teigi{\iLexparse{}}~\cite[Def.~11]{navarro21approximation} is a factorization $T = F_1 \cdots F_v$ such that $F_x = T[\dst_x \IC{} \dst_x + \ell_x -1]$ 
with $\dst_1 = 1$ and $\dst_{x+1} = \dst_x + \ell_x$ if $\ell_x := \PLCP[\dst_x] \ge 1$,
or $F_x$ is a \teigi{literal factor} with $\ell_x := |F_x| = 1$ otherwise.
For $\PLCP[\dst_x] \ge 1$, the reference of $F_x$ is $\Phi(F_x)$ since by definition of the PLCP array
$T[\dst_x \IC{} \dst_x + \ell_x -1] = T[\Phi(\dst_x) \IC{} \Phi(\dst_x) + \ell_x -1] $.
See \cref{figFactLexparse} for an example.
\J{Among all factorizations based on the selection of a lexicographically preceding position as factor reference, \iLexparse{} produces the least number of factors~\cite{navarro21approximation} and is therefore a lower bound for other factorizations based on the lexicographic order such as lcpcomp~\cite{dinklage17tudocomp}
 or plcpcomp~\cite{dinklage19plcpcomp}.
}

\teigi{Closed factorizations}~\cite{badkobeh16closed} create factors that are \teigi{closed};
A factor $F$ is closed if $F$ has a proper substring that is both a prefix and a suffix of $F$ but nowhere else appears in $F$.
\citeauthor{badkobeh16closed} studied two such factorization.
First, the \teigi{longest closed factorization} factorizes $T$ into $T = F_1 \cdots F_z$ such that $F_x$ is the \emph{longest} prefix of
$T[|F_1|\cdots |F_{x-1}|+1 .. n]$ that is closed.
The longest closed factorization of $T = \mathtt{\runningExample{}}$ is $T = F_1 \cdot F_2 = \texttt{ababbababbab} \cdot \texttt{b}$,
where the border \texttt{ababbab} of $F_1$ nowhere else appears in $F_1$.
Second, the \teigi{shortest closed factorization} factorizes $T$ into $T = F_1 \cdots F_z$ such that $F_x$ is the \emph{shortest} prefix of 
$T[|F_1|\cdots |F_{x-1}|+1 .. n]$ that is closed and is of length at least 2.
The shortest closed factorization may not exist, for instance if $T$ starts with a unique character.
The shortest closed factorization of $T = \mathtt{\runningExample{}a}$ is $T = \texttt{aba} \cdot \texttt{bb} \cdot \texttt{aba} \cdot \texttt{bb} \cdot \texttt{abba}$.

\subsection{Lexparse}\label{secLexparse}

\J{Despite being perceived that both $\Phi$ and $\PLCP$ seem necessary for a linear-time computation~\cite{koppl22computing},
it actually suffices to have sequential access to $\Phi$ and random access to the text.
}
That is because we can naively compute $\PLCP[\dst] = \lce(\dst, \Phi(\dst))$ in $\Oh{{\PLCP[\dst]}}$ time by linearly scanning the character pairs in the text.
The number of scanned character pairs sum up to at most $2n$.
This however scales linearly in the interval length.
Here, our idea is to use LCE queries to speed up the factor-length computation.
To find the reference position, i.e., the position from where we want to compute the LCE queries, we use the following data structure.

\begin{lemma}[{\cite[Theorem~3.8]{babenko15wavelet}}] \label{thmInternalWavelet}
There is a data structure that can give, for an interval~\intervalI{},
the $k$-th smallest suffix of $T[\intervalI]$ or the rank of a suffix of $T[\intervalI]$, 
each in \Oh{\lg s} time, where $s = |\intervalI|$.
It can be constructed in $\Oh{n \sqrt{\lg n}}$ time, and takes \Oh{n \lg n} bits of space.
\end{lemma}

Having the data structure of \cref{thmInternalWavelet} and \Oh{1}-time support for LCE queries,
we can compute the substring compression variant of \iLexparse{} for a given query interval $\intervalI$
with the following instructions in \Oh{\VSub \lg s} time, 
where  $s = |\intervalI|$ and
$\VSub$ is the number of factors of the $\iLexparse$ factorization of $T[\intervalI]$.

\myblock{Algorithm}
Start at text position $\ibeg{\intervalI}$, and query the rank of the substring suffix $T[\ibeg{\intervalI}..\iend{\intervalI}]$.
Given this rank is $k$, select the $(k-1)$-st substring suffix of $T[\intervalI]$.
Say this suffix starts at position~$j$, then $j$ is the reference of the first factor.
Next, compute $\ell \gets \lce(\ibeg{\intervalI}, j)$ to determine the factor length~$\ell$.
On the one hand, if the computed factor protrudes the substring $T[\intervalI]$ with $\ibeg{\intervalI} + \ell - 1 > \iend{\intervalI}$, trim $\ell \gets \iend{\intervalI} - \ibeg{\intervalI} + 1$.
On the other hand, if $k = 1$ or $\ell = 0$, then the factor is literal.
Anyway, we have computed the first factor, and recurse on the interval~$[\ibeg{\intervalI}+\ell..\iend{\intervalI}]$ to compute the next factor while keeping the query interval $\intervalI$ fixed for the data structure of \cref{thmInternalWavelet}.
Each recursion step in the algorithm uses a select and a rank query on the data structure of \cref{thmInternalWavelet}, which takes $\Oh{\lg s}$ time each.

\begin{theorem}\label{thmLexparse}
There is a data structure that,
given a query interval $\intervalI$,
can compute lexparse of $T[\intervalI]$ in $\Oh{z \lg s}$ time, where $z$ is the number of lexparse factors and $s = |\intervalI|$ the length of the query interval.
This data structure can be constructed in $\Oh{n \sqrt{\lg n}}$ time, 
using $\Oh{n \lg n}$ bits of working space.
\end{theorem}

\begin{algorithm}
  \caption{Longest closed factorization with the bounds stated in \cref{thmClosedFact}.}
  \label{algoLongestClosed}
	\begin{algorithmic}[1]
  \Require query interval $\intervalI$ 
  \If{$\intervalI = \emptyset$}
  \Return
  \EndIf
  \Comment{factorization of empty interval is empty}
  \State $\lambda \gets \fnSelectLeaf(\ibeg{\intervalI})$ 
  \Comment{$\Oh{\timeSA}$ time}
	\label{lineCloseLambdaOne}
  \State 
  \(
    u \gets \argmax_{u \text{~ancestor of~} \lambda} \{ \fnDepth(u) \mid \SA{}.\RMQ[\lmostleaf(u)..\rmostleaf(u)] \neq i \}.
  \)
  \If{$u = $ root}
  \Comment{find $u$ via binary search on $d \mapsto \levelanc(d,\lambda)$}
	\label{lineCloseUroot}
  \State \Output factor $T[\ibeg{\intervalI}]$ 
  \Comment{$\ibeg{\intervalI}$ is rightmost occurrence of character $T[\ibeg{\intervalI}]$}
  \State \Return by recursing on $[\ibeg{\intervalI}+1..\iend{\intervalI}]$
  \EndIf
  \State $\ell \gets \strdepth(u)$
  \Comment{invariant: if $u \neq$ root then $\ell \ge 1$; $j = $ range successor of $\ibeg{\intervalI}$}
  \State $j \gets \SA{}.\rangeNextValue{}(\ibeg{\intervalI}, [\lmostleaf(u)..\rmostleaf(u)])$
  \Comment{$\Oh{\timeSA + \lg n}$ time}
\If{$j+\ell \le \iend{\intervalI}$}
	\label{lineCloseUend}
  \State \Output factor $T[\ibeg{\intervalI}..j+\ell-1]$
  \State \Return by recursing on $[j+\ell..\iend{\intervalI}]$
  \EndIf
  \LComment{the computed factor protrudes $\intervalI \Rightarrow$ cut its right border}
  \State
  \(
    u' \gets \argmin_{u' \text{~ancestor of~} u} \{ \fnDepth(u') \mid \strdepth(u') \ge \iend{\intervalI} - j + 1 \}.
  \)
	\label{lineCloseUcap}
    \State $\ell' \gets \min(\strdepth(u'), \iend{\intervalI}-j+1)$
    \LComment{even if we cut the right border, by definition of $u'$ it cannot appear elsewhere in the factor}
    \State $j' \gets \SA{}.\rangeNextValue{}[\ibeg{\intervalI}, \lmostleaf(u)..\rmostleaf(u)]$
  \Comment{$\Oh{\timeSA + \lg n}$ time}
    \State \Output factor $T[\ibeg{\intervalI}..j'+\ell'-1]$
    \State \Return by recursing on $[j'+\ell'..\iend{\intervalI}]$
  \end{algorithmic}
\end{algorithm}

\subsection{Closed Factorizations}\label{secClosedFact}
For the shortest closed factorization, \citet[Lemma 3]{badkobeh16closed} showed 
that a shortest closed factor of length at least two has a unique border of length one, 
where a \teigi{border} of a string $X$ is a substring of $X$ that is both a prefix and a suffix of $X$.
Given that we have factorized a prefix of $T$ by $F_1 \cdots F_{x-1}$ and our task is to compute $F_x$ starting at $\src_x := 1 + \sum_{y=1}^{x-1} |F_y|$,
we scan for the succeeding occurrence of $T[\src_x]$ in $T[\src_x+1..]$, which determines the end of $F_x$
because then $F_x$ is the shortest substring starting at $\src_x$ having $T[\src_x]$ as a border appearing nowhere else in $F_x$.
By preprocessing $T$ with a wavelet tree, we can \emph{rank} the current occurrence of $T[\src_x]$ and 
\emph{select} the subsequent occurrence of this character.
There are implementations~\cite[Theorem~2]{barbay14wavelet}\cite[Theorem~4]{golynski08redundancy} achieving time $\Oh{1 + \min\left(\frac{\lg \sigma}{\lg \lg n}, \lg \lg \sigma\right)}$ per query if we reduce the alphabet of $T$ such that each character of the alphabet appears in $T$.
Now, given a queried interval $\intervalI \subset [1..n]$, we can factorize $T[\intervalI]$ with $\Ot{z}$ rank/select queries,
where $z$ is the number of computed factors.
It is also possible to store for each character $c \in \Sigma$ a bit vector $B_c[1..n]$ marking all occurrences of $c$ in $T$, and augment $B_c$ with rank/select support data structures that answer each query in constant time,
but totaling up to $n\sigma + \oh{n\sigma}$ bits of space, which is better than a naive solution storing all answers when $\sigma \in \oh{n}$.

For computing the longest closed factorization, we only slightly modify the algorithm presented in \cite[Section 3]{badkobeh16closed}.
The idea there is to use the suffix tree (as we will also later for the LZ78-related factorizations).
In detail, they precompute an array $P[1..n]$ such that $P[i]$ stores the highest ancestor of the leaf $\lambda$ with suffix number $i$ among all ancestors $v$ of $\lambda$ with the property that the largest suffix number of all leaves in the subtree rooted in $v$ is $i$, i.e.,
\[
  v = \argmin_{v \text{~ancestor of~} \lambda} \{ \fnDepth(v) \mid \SA{}.\RMQ[\lmostleaf(v)..\rmostleaf(v)] = i \}.
\]
Like the authors explained, computing $P$ can be done by a preorder traversal of the suffix tree.
Now, given a queried interval $\intervalI \subset [1..n]$, we can factorize $T[\intervalI]$ with $P$ and the suffix tree, cf.~\cref{algoLongestClosed}.
Namely, we take the node $P[\ibeg{\intervalI}]$ corresponding to the leaf $\lambda$ with suffix number $\ibeg{\intervalI}$ (Line~\ref{lineCloseLambdaOne}).
Let $P[\ibeg{\intervalI}]$'s parent be $u$. 
By the property of $P$, $u$ is the \emph{lowest} ancestor of $\lambda$ whose subtree has a leaf with a suffix number $j$ greater than $\ibeg{\intervalI}$.
If $u$ is the root, then $\ibeg{\intervalI}$ is the rightmost occurrence of character $T[\ibeg{\intervalI}]$ in $T$.
Consequently, the longest closed factor starting at text position~$\ibeg{\intervalI}$ is just $T[\ibeg{\intervalI}]$ (Line~\ref{lineCloseUroot}).

Otherwise, the string label $L$ of $u$ is not empty.
$L$ occurs at least as two substrings, two of them starting at text positions $\ibeg{\intervalI}$ and $j$.
The substring $T[\ibeg{\intervalI}..j+|L|-1]$ is therefore bordered with border~$L$, but may not be necessarily closed because $L$ may appear inside somewhere else.
For that property to hold, we need to select the \emph{smallest} such $j$, i.e., 
the range successor of $\ibeg{\intervalI}$ among the suffix numbers in the subtree of $u$.
Finally, we need to check whether $j+|L| \le \iend{\intervalI}$ (Line~\ref{lineCloseUend}).
If this does not hold, we need to shorten the border~$L$ that protrudes the query interval $\intervalI$ (even the case that $j > \iend{\intervalI}$ may hold and needs to be considered).
However, by shortening the border, the shortened border~$L'$ may now appear inside the computed factor elsewhere.
As a remedy, we need to update $j$ to the leftmost occurrence of $L'$ after $\ibeg{\intervalI}$.
To this end, we visit the shallowest ancestor~$u'$ with a string depth of at least $j-\iend{\intervalI}$, which can $u$ be itself (Line~\ref{lineCloseUcap}).
If $u'$ is now the root, we are at the first case above.
Otherwise, with $u'$ we redetermine $j$, which can be now further to the left, but we are guaranteed that the shortened border does not appear elsewhere.

\block{Time}
For efficient computation, we need a data structure for finding the range successor~$j$.
Remembering that reading the suffix numbers of the suffix tree leaves in leaf-rank order gives the suffix array,
we can build a wavelet tree on the suffix array to compute~$j$.
The wavelet tree can answer find $j$ in $\Oh{\lg n}$ time by a query coined as \rangeNextValue{}~\cite[Theorem~7]{gagie12new}.
\begin{itemize}
  \item $\SA{}.\rangeNextValue{}(x, \intervalJ)$ returns the successor of $x$ in $\SA[\intervalJ]$, i.e.,
\end{itemize}
   \[ 
    \min \{ \SA[x] ~\mid~ x \in \intervalJ \wedge \SA[x] > v \}
  \]

Next, to determine $u'$, we want a data structure that finds ancestors based on the string depth, and not just on depth like $\strdepth$.
To this end, we make use of the following data structure.

\begin{lemma}[\cite{gawrychowski14weighted,belazzougui21weighted}] \label{lemWeightedAncestor}
  There exists a weighted ancestor data structure for \ST{}, 
  which supports, given a leaf~$\ell$ and integer~$d$, 
  constant-time access to the ancestor of $\ell$ with string depth~$d$.
  It can be constructed in linear time using $\Oh{n \lg n}$ bits of space.
\end{lemma}
In total, we can compute the factorization of $T[\intervalI]$ in $\Oh{z \lg n}$ time, where $z$ is the number of computed factors.

\begin{example}
For our running example $T = \mathtt{\runningExample}$ with $\intervalI = [1..n]$, $P[1]$ is the leaf with suffix number $1$ itself. 
Its parent has preorder number~$4$ with string length~7; the succeeding suffix starts at position $6$. 
Hence, the first factor is $T[1..6+7-1] = \texttt{ababbababbab}$.
\end{example}

\block{Reducing Space}
If we allow us more time, we can spend some time to compute the $P$ values on the fly.
Instead of storing $P$, we perform a binary search with $d \mapsto \levelanc(\ell,d)$ to find the node $P[\ibeg{\intervalI}] = v$.
The binary search takes the maximum suffix number in the ancestor node's subtree as search key, i.e., $\SA{}.\RMQ[\lmostleaf(v)..\rmostleaf(v)]$.
These values form a non-decreasing monotonic sequence for increasing values of $d$, and therefore a binary search is feasible.
Finding $v$ therefore takes $\Oh{\timeSA \lg n}$ time. 
Similarly, for finding $u'$ we drop the weighted ancestor data structure of \cref{lemWeightedAncestor}, but perform the aforementioned binary search with the evaluation of $\strdepth$ as key values.
The search for $u'$ takes similarly $\Oh{\timeSA \lg n}$ time.
Both binary searches become now the time bottleneck per factor since we additionally have
$\Oh{\lg n}$ time for \rangeNextValue{} and $\Oh{\timeSA}$ time for $\strdepth(u)$ to determine $L$.
The final algorithm is given in \cref{algoLongestClosed}.
We obtain the following result:

\begin{theorem}\label{thmClosedFact}
There is a data structure that,
given a query interval $\intervalI$,
can compute the longest or shortest closed factorization of $T[\intervalI]$ in (a) $\Oh{z \lg n}$ time or (b) $\Oh{z \lg^{1+\epsilon} n}$ time, 
where $z$ is the number of lexparse factors.
This data structure can be constructed in $\Oh{n}$ time, 
using (a) $\Oh{n \lg n}$ bits or (b) $\Oh{n \lg \sigma}$ bits of working space.
\end{theorem}

In what follows, we apply the suffix tree traversal techniques to other factorizations.
Two of them belong to the family of semi-greedy parsings, which we introduce next.

\section{Semi-Greedy Parsing}\label{secSemiGreedy}
The semi-greedy parsing~\cite{hartman85optimal} is a variation of the LZ77 parsing in that 
a factor $F_x$ of the factorization $F_1 \cdots F_z$ starting at text position~$\dst_x := |F_1 \cdots F_{x-1}|$ does not necessarily have to be the longest prefix of $T[\dst_x..]$ that has an occurrence starting before $\dst_x$.
Given that the longest prefix has length $\ell$, 
the semi-greedy parsing instead selects the length $\ell' \in [1..\ell]$ that maximizes the sum of the lengths of the current and the next factor $|F_x| + |F_{x+1}|$ like $F_x$ would have length $\ell'$ and $F_{x+1}$ would have been selected by the standard greedy parsing continuing at $T[\dst_x+\ell'..]$.

\myblock{Semi-Greedy LZ77}
\citet[Algorithm~10]{crochemore12lpf} proposed an algorithm working with the array \LPF{} to compute the semi-greedy parsing in \Oh{n} time. 
$\LPF$, called longest previous factor table~\cite{crochemore08lpf},
is an array of length $n$ with
$\LPF[j] = \max \{\ell \mid \text{there exists an~} i \in [1..j-1] \text{~such that~} T[i..i+\ell-1] = T[j..j+\ell-1] \}$ for every $j \in [1..n]$.
The algorithm of \citet{crochemore12lpf} can be modified such that a substring query for $\intervalI$ with $1 \in \intervalI$ can be answered in $\Oh{z}$ time if $\Oh{n}$ preprocessing is allowed.
For that, build an RMQ data structure on $f : i \mapsto \LPF[i] + i$ for all $i \in [1..n]$
such that the query for the range $R := [i..i+\LPF[i]]$ leads to the position $i+\ell'-1 \in R$ for which $f(i)$ is maximal.
Thus, we can spend constant time per factor; we stop when the last computed factor~$F_z$ ends or goes beyond the end of \intervalI{}, where we trim $F_z$ in the latter case.
For supporting intervals that do not start with the first text position, we need to recompute $\LPF$ whose entries depend on previous substring occurrences.

\myblock{Semi-Greedy LZW}
\citet{horspool95effect} proposed adaptations of the semi-greedy parsing for Lempel--Ziv--Welch (LZW)\J{~\cite{welch84lzw}}, a variation of the LZ78 factorization.
These LZ78-based semi-greedy parsings have been studied by \citet{matias99optimality}, who further generalized the semi-greedy parsing to other parsings, and coined the name \teigi{flexible parsing} for this parsing strategy.
They also showed that the flexible parsing variant of parsings using prefix-closed dictionaries is optimal with respect to the minimal number of factors.
A set of strings $\mathcal{S} \subset \Sigma^+$ is called \teigi{prefix-closed} if every non-empty prefix of every element of $\mathcal{S}$ is itself an element of $\mathcal{S}$.
A follow-up~\cite{matias01effect} presents practical compression improvements as well as algorithmic aspects in how to compute the flexible parsing of LZW, called \iLZWFP{}, or the variant of \citet{horspool95effect}, which they call \iFPA{}.
The main difference between \iFPA{} and \iLZWFP{} is that \iLZWFP{} uses the parsing dictionary of the standard LZW factorization, 
while the factors in \iFPA{} refer to substrings that would have been computing greedily at the starting positions of the \iFPA{} factors.
The selection of reference makes the computation of \iLZWFP{} (theoretically) easier since we can statically build the LZW parsing dictionary with a linear-time LZW construction algorithm like~\cite{fischer18lz}.
It also lets the factors of \iLZWFP{} be upper bounded by the factors of LZW because we do not change the original LZW dictionary while making a wiser choice for the factors.
On the other hand, \iFPA{} may create more factors than LZW, which has however not yet been observed in practice.
The major conceptional difference to LZW is that a \teigi{reference} is no longer a factor of the computed factorization itself.
Instead, a reference is a substring of the prefix-free dictionary and has a \teigi{reference index} mimicking the definition \emph{factor index}, which together with \emph{referred factor} is obsolete in this section in favor for the terms \emph{reference} and \emph{reference index}.
References are assigned their indices based on the time of creation. 
Because the dictionary is prefix-free, it exhibits the same characteristics as the LZ trie, which now represents the references instead of the factors.
For \iLZWFP{}, the LZ trie is actually the LZ trie of LZW.
In what follows, we adapt both flexible parsing variants to LZ78 as a means of didactic reduction (the adaptation from LZ78 to LZW is straightforward).
We first formally define both factorizations in the LZ78 terminology and subsequently give algorithms for the factorizations.

\myblock{\iLZEightFP{}: LZ78 version of \iLZWFP{}}
Briefly speaking, we apply the flexible parsing to the LZ78 factorization without changing the LZ78 dictionary at any time. Thus, for computing a factor $F_x$, we restrict us to refer to any LZ78 factor that ends prior to the start of $F_x$.
The rationale of this factorization is that the decompression works by building the LZ78 dictionary while decoding the factors of the flexible parsing.
For a formal definition, assume that the LZ78 factorization of $T$ is $T = R_1 \cdots R_{z_{78}}$.
Given $R_0$ denotes the empty string,
we write $\iend{R_0} = 0$ and $\iend{R_x} := \sum_{y=1}^{y=x} |R_y|$ for the ending position of $R_x$ in $T$.
Now suppose we have parsed a prefix $T[1..i-1] = F_1 \cdots F_{x-1}$ with \iLZEightFP{} and want to compute $F_x$.
At that time point, the dictionary $\DICT$ consists of all LZ78 factors that end before $\dst_x$, appended by any character in $\Sigma$. 
Instead of greedily selecting the longest match in $\DICT$ for $F_x$ as LZ78 would do,
we select the length $|F_x|$ for which $F_x \cdot F_{x+1}$ is longest (ties are broken in favor of a longer $F_x$). 
Note that at the time for selecting $F_{x+1}$, we can choose among all LZ78 factors that end before $\dst_{x+1}$, which depends on $|F_x|$.
Formally, let $y$ be the index of the longest LZ78 factor~$R_y$ that is a prefix of the starting position $T[\dst_x..]$ of $F_x$, i.e., 
\[
  y = \argmax_{y' \in [1..z_{78}]}  \{ \abs{R_{y'}} :  R_{y'} = T[\dst_x..\dst_x+|R_{y'}|-1] \wedge \iend{R_{y'}} < \dst_x \}.
\]
Then we determine in the interval $\intervalJ := [\dst_x+1..\dst_x+|R_y|]$ the starting position of $F_{x+1}$
by the flexible parsing, i.e., we select the position $p \in \intervalJ$ for which the value
\[
  \argmax_{p \in \intervalJ} \{ p + \abs{R_{y'}} : R_{y'} = T[p..p+|R_{y'}|-1] \wedge \iend{R_{y'}} < p \}
\]
is maximized.
Then $F_x$ has length $p-\dst_x$; if $|F_x| \ge 2$, then its reference exists since LZ78 is prefix-closed.
By doing so, we obtain the factorization of \iLZEightFP{}.

\myblock{\iLZEightFPA{}: LZ78 version of \iFPA{}}
The other semi-greedy variant we study in its LZ78 version is the one of \iFPA{}, which we call \iLZEightFPA{}.
The only difference to \iLZEightFP{} is that the dictionary \DICT{} is based on the LZ78 factors $R'_1, \ldots, R'_{z'_{78}}$ starting at the \iLZEightFPA{} factor starting positions $\dst_x$ for all $x$.
The definition of $\mathcal{J}$ and $p$ is identical to above be switching $R_y$ with $R'_y$.
Similarly, our task is to select the length~$|F_x| \in [1..|R'_y|+1]$ such that we advance the furthest in the text with a factor starting at $\dst_x+|F_x|$, for all such possible lengths $|F_x|$.

\begin{figure}[htpb]
  \centering
  \includegraphics[width=0.3\textwidth,page=1]{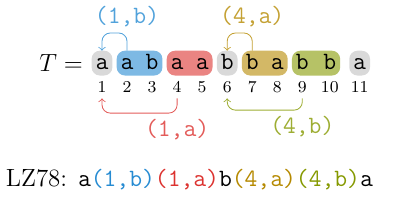}
  \includegraphics[width=0.3\textwidth,page=2]{./factorization/flex}
  \includegraphics[width=0.3\textwidth,page=3]{./factorization/flex}
  \caption{Visualization of the factorizations in \cref{exFP}.}
  \label{figExFlex}
\end{figure}

\begin{example}\label{exFP}
  Let $T = \texttt{aabaabbabba}$. The LZ78 factorization of $T$ is
  $
  R_1 = \texttt{a},
  R_2 = \texttt{ab} = R_1 \texttt{b},
    R_3 = \texttt{aa} = R_1 \texttt{a},
    R_4 = \texttt{b},
    R_5 = \texttt{ba} = R_4 \texttt{a},
    R_6 = \texttt{bb} = R_4 \texttt{b},
    R_7 = \texttt{a}$.
    \iLZEightFP{} builds upon $R_1 \cdots R_7$ to produce the factorization
    $
    F_1 = \texttt{a},
    F_2 = \texttt{ab} = R_1  \texttt{b},
    F_3 = \texttt{a},
    F_4 = \texttt{abb} = R_2 \texttt{b},
    F_5 = \texttt{abb} = R_2 b,
    F_6 = \texttt{a}.
    $
    Noteworthy is the choice of $F_3$ being shorter than $R_3$ in the favor for creating a longer factor $F_4$ contrary to the 4-th factor $R_4$ of LZ78.

    Finally, the \iLZEightFPA{} factorization is 
    $
    F'_1 = \texttt{a},
    F'_2 = \texttt{ab} = R'_1 \texttt{b},
    F'_3 = \texttt{a},
    F'_4 = \texttt{abb} = R'_2 \texttt{b},
    F'_5 = \texttt{abba} = R'_4 \texttt{a},
    $
    where $R'_1 = \texttt{a}$ starts at $T[1..]$, $R'_2 = \texttt{ab}$ at $T[2..]$, 
    $R'_3 = \texttt{aa}$ at $T[4..]$,
    $R'_4 = \texttt{abb}$ at $T[5..]$.
    Noteworthy is the overlapping of $R'_3$ and $R'_4$, 
    which is because $R'_j$ starts with $F'_j$ for every $j$ by definition.
    This choice makes it possible to let $F'_5$ refer to the preferably long reference $R'_4$.
    We thus could shorten the factorization size from 7 (LZ78) to 6 (\iLZEightFP{}) and 5 (\iLZEightFPA{}) with the flexible parsing variants.
    A visualization of all three factorizations is given in \cref{figExFlex}.
\end{example}

The difference on the number of factors can be larger than the experienced constant in \cref{exFP}.
In factor, the following example gives a difference in $\Ot{n/\lg n}$.

\begin{example}
  Select an LZ78-incompressible string $S$ of length $n$ on the alphabet $\{\texttt{b}, \texttt{c}\}$ such that the number of LZ78 factors of $S$ is $\Ot{n/\lg n}$.
Suppose we parse the string $T = \texttt{a} \cdot \texttt{a} S[1] \cdot \texttt{a} S[1..2] \cdot \texttt{a} S[1..3] \cdots \texttt{a} S[1..n] \cdot \texttt{a} S[1..n] \cdot \texttt{a} \cdot S[1..n]$.
On the one hand, the LZ78 factorization would create the factors $R_1 = \texttt{a}$, $R_2 = \texttt{a} S[1]$, $R_3 = \texttt{a} S[1..2]$, $\ldots$, $R_{n+1} = \texttt{a} S[1..n]$,
$R_{n+2} = \texttt{a} S[1..n] \texttt{a}$, and finally create the $\Ot{n/\lg n}$ factors for factorizing $S$.
On the other hand, both flexible parsing variants have $F_x = R_x$ for all $x \in [1..n+1]$, but select $F_{n+2} = R_{n+1} = F_{n} S[n]$ because doing so 
maximizes the length of $F_{n+3} = \texttt{a} S[1..n] = R_{n+1} = F_n S[n]$.
While the flexible parses have thus $n+3$ factors, the LZ78 factorization has $n + 2 + \Ot{n / \lg n}$ many.
\end{example}

\myblock{Decompression}
By storing the \iLZEightFP{} factors of $T$ as a list of pairs like in LZ78 allows for reconstructing~$T$.
While literal factors can be restored by reading the stored characters, 
a referencing factor refers to a reference that is an LZ78 factor.
We therefore need to compute the LZ78 factorization on the decoded output on-the-fly.
For \iLZEightFPA{}, for each referencing factor we read we need to restore an LZ78 factor starting at the same position, which becomes a reference.
This is not immediately possible if the read \iLZEightFPA{} factor is shorter, so we trigger this LZ78 factor computation when enough text has been decompressed.
This is also the reason why, at the computation of an \iLZEightFP{} or \iLZEightFPA{} factor~$F_x$, 
we can only choose references representing factors that \emph{end before} the starting position $\dst_x$ of $F_x$.

\subsection{Computation with AC automaton}\label{secACautomaton}
For computation of any of the two LZ78 flexible parsings, let us stipulate that we can compute the longest possible factor starting in $\intervalJ$ for each text position by a lookup in \DICT{} within $\timeDICT$ time.
\citet[Section~3]{matias01effect} achieved $\timeDICT = \Oh{1}$ expected time per text position by storing 
Karp--Rabin fingerprints\J{~\cite{karp87efficient}} in a hash table to match considered substrings with the parsing dictionary that represents its elements via fingerprints.
Alternatively, they get constant time for constant-sized alphabets with two tries, where one trie stores the dictionary, and the other the reversed strings of the dictionary.
Despite claimed~\cite{matias99optimality} that \iLZWFP{}, the flexible parsing of LZW, can be computed in linear time, no algorithmic details have been given.

However, it is also possible to set $\timeDICT$ to the time for a node-to-child traversal time in a trie implementation.
For \iLZEightFP{}, we build an Aho--Corasick (AC) automaton~\cite{aho75efficient}
on \DICT{} just after we have computed the LZ78 factorization, 
so \DICT{} stores all (classic) LZ78 factors~$R_1, \ldots, R_{z_{78}}$.
Roughly speaking, this AC automaton is the LZ trie augmented with suffix links.
Like the LZ trie, the AC-automaton allows us to search for the longest matching prefix of a suffix $T[p..]$ in $\Oh{\ell \timeDICT}$ time if $\ell$ is the length of this prefix,
where $\timeDICT$ is the time for a node-to-child traversal.
To this end, we traverse the LZ trie downwards, matching characters of $T[p..]$ with edge labels.
Given we read $\ell$ characters and end up at a node $v$ such that none of its out-going edges matches with $T[p+\ell]$,
then $T[p..p+\ell-1]$ is the longest LZ78 factor that is a prefix of $T[p..]$.
Given this factor is $R_w$,
we say that $v$ is the \teigi{locus} of $R_w$, or more generally: the locus of an LZ78 factor $F$ is the LZ trie node whose string label is $F$.

With suffix links, the AC automaton allows us to select, from the locus of $R_w$ in the trie, 
the longest proper suffix $T[p'..p+\ell-1]$ with $p' > p$ of the matched LZ78 factor~$R_w$ that is again an LZ78 factor~$R_u$.
Since \DICT{} is prefix-closed, for all other LZ78 factors that would match with prefixes $P$ of $T[k..]$ with $k \in [p+1..p'-1]$, $P$ is shorter than $R_u$, and can be omitted for the computation of $F_x$.
Nevertheless, we need to traverse downwards from the locus of $R_u$ to find the longest possible LZ78 factor that is a prefix of $T[p'..]$.
The explained algorithm uses the following two queries on the AC automaton:
\begin{itemize}
  \item $\fnSuffixLink(w)$: returns the node whose string label is longest among all proper suffixes of $w$ that are string labels of nodes of the LZ trie.
  \item $\fnChild(v,c)$: returns the child of the LZ trie node $v$ connected by an edge with label $c \in \Sigma$.
\end{itemize}
Our algorithm computes the first search window $T[p..p+\ell-1]$ by $\ell+1$ $\fnChild{}$ queries, but then subsequently 
uses $\fnSuffixLink$ to shrink the window from the left-hand side or extend the window on the right-hand side by $\fnChild$.
The total number of $\fnSuffixLink$ and $\fnChild$ queries is thus bounded by the number of characters we want to factorize.

To correctly compute $F_x$, we need to omit nodes that correspond to factors $R_y$ with $\iend{R_y} \ge \dst_x$. 
To this end, we augment the node corresponding to $R_y$ with $\iend{R_y}$, for each $y$.
When traversing the automaton for pattern matching to find the longest match of $T[p..]$ in $\DICT$, we stop the traversal as soon as we try to visit a node with factor index $y$ and $\iend{R_y} \ge \dst_x$.
That is because the factor index of a child is larger than of its parent, and thus the augmented values (i.e., the LZ78 factor indices) strictly grow when we descend in the AC automaton.
Further, while matching $T[\dst_x+\ell..]$, we may take a suffix link to a node $v$ with factor index $y$ and $\iend{R_y} \ge \dst_x$; in this case we skip $v$ and continue with taking the suffix link of $v$ (which involves increasing $\ell$).

Finally, to perform a query in the AC automaton only once per text position, we temporarily memoize the computed references for the range $[\dst_{x+1}..\dst_x+|R_w|] \subset \intervalJ$ such that we can skip the recomputation of the longest match when searching for the length of $F_{x+1}$. (Recall that $R_w$ is the longest LZ78 factor starting at $\dst_x$ with $\iend{R_w} < \dst_x$.)
With the memoization we compute the same factorization since creating factors does \emph{not} change the parsing dictionary.
The memoized values are discarded after having factorized the respective text positions.
The number of memoized positions is upper-bounded by the longest LZ78 factor length, which we can bound by $\Oh{\sqrt{n}}$.

\begin{theorem}
  We can compute \iLZEightFP{} in 
  \Oh{n \timeDICT} time with $\Oh{(z+\delta) \lg n}$ bits of space,
  where $\delta = \Oh{\sqrt{n}}$ is the number of memoized values.
  If we answer \fnChild{} queries with balanced binary search trees, then
  $\timeDICT = \Oh{\lg \sigma}$.
\end{theorem}

To adapt this approach for \iLZEightFPA{}, we need a dynamic dictionary, which gets filled with the LZ78 factors $R'_1, R'_2, \ldots$ starting with the \iLZEightFPA{} factors.
It is possible to make the AC-automaton semi-dynamic in the sense that it supports adding new strings to \DICT{},
for which we are aware of two approaches~\cite{meyer85incremental,hendrian19efficient} whose running times depend on the actual strings indexed by the AC-automaton.

\subsection{Computation with Suffix Trees}\label{secFlexParseLinTime}

The algorithms we propose in the following are based on 
the suffix tree superimposition with the LZ trie of~\cite{fischer18lz}, which we briefly review.

\myblock{Suffix Tree Superimposition}
An observation of \citet[Sect.~3]{nakashima15position} is that 
the {LZ} trie is a connected subgraph of the suffix trie containing its root.
That is because the LZ78 parsing dictionary (or the dictionary of any variant of our semi-greedy parsings) is prefix-closed.
We can therefore simulate the {LZ} trie by marking nodes in the suffix trie.
Since the suffix trie has \Oh{n^2} nodes, we use the suffix tree~\ST{} instead of the suffix trie to save space.
In \ST{}, however, not every {LZ} trie node is represented;
these implicit {LZ} trie nodes are on the \ST{} edges between two \ST{} nodes.
Since the LZ trie is a connected subgraph of the suffix trie sharing the root node,
implicit {LZ} trie nodes on the same \ST{} edge have the property that they are all consecutive and that the first one starts at the first character of the edge.
To represent all implicit LZ trie nodes, it suffices to augment each \ST{} edge~$e$ with a counter counting the number of $e$'s implicit {LZ} trie nodes.
We call this counter an \teigi{exploration counter}, 
and we write~$n_v \in [0\IC{}|e|]$ for the exploration counter of an edge~$e = (u,v)$, which is stored in the lower node~$v$ that $e$ connects to.
Here, $|e|$ denotes the length of $e$'s label.
We say the \teigi{locus} of a substring $S$ in $T$ is the highest \ST{} node whose string label has $S$ as a prefix.
Thus, the locus of an LZ factor $F$ is a node $v$ whose string label is $F$ (then $F$ has an explicit LZ trie node represented by $v$) or contains $F$ as a proper prefix (then $F$ is represented implicitly by the exploration counter~$n_v$).
Additionally, we call an \ST{} node~$v$ an \teigi{edge witness} if $n_v$ becomes incremented during the factorization.
We say that $n_v$ is \teigi{full} if $n_v$ is the length of the string label of the edge connecting to $v$,
meaning that $v$ is an explicit {LZ} trie node.
We additionally stipulate that the root of \ST{} is an edge witness, whose exploration counter is always full.
Then all edge witnesses form a connected sub-graph of \ST{} sharing the root node.
See \cref{figSuffixTreeImposition} for an example.
Finally, \citet[Sect.~4.1]{fischer18lz} gave an \Oh{n}-bits representation of all exploration counters, which however builds on the fact that the parse dictionary is prefix-closed.

\myblock{Linear-Time Computation}
To obtain an \Oh{n \timeSA}-time deterministic solution for \iLZEightFP{} and \iLZEightFPA{}, we make use of the superimposition of \ST{}.
For \iLZEightFP{}, compute the LZ78 factorization as described in \cite[Sect.~4]{fischer18lz} such that all edge witnesses are marked.
The marking is done with the following data structure.

\begin{lemma}[\cite{cole05dynamic}] \label{lemLowestMarkedAncestor}
There is a semi-dynamic lowest common ancestor data structure that can 
(a) find the lowest marked node of a leaf or 
(b) mark a specific node, 
both in constant time.
We can augment \ST{} with this data structure in \Oh{n} time using \Oh{n \lg n} bits of space.
\end{lemma}

\begin{algorithm}
\caption{Computing \iLZEightFPA{}.}
  \label{algoFPA}
  \begin{algorithmic}[1]
  \Require query interval $\intervalI$, root is marked
  \If{$\intervalI = \emptyset$}
  \Return
  \EndIf
  \State $\lambda_1 \gets \fnSelectLeaf(\ibeg{\intervalI})$ 
  \Comment{$\Oh{\timeSA}$ time}
	\label{lineFPALambdaOne}
  \State $w_1 \gets \text{~lowest marked ancestor of~} \lambda_1$
	\label{lineFPAAncWOne}
  \If{$w_1 = $ root}
  \Comment{factor is literal}
  \State \Output factor $T[\ibeg{\intervalI}]$ 
  \State \Return by recursing on $[\ibeg{\intervalI}+1..\iend{\intervalI}]$
  \EndIf
  \State $\ell_0 \gets \strdepth(\parent(w_1))+n_w+1, p \gets 0, \ell_{\text{max}} \gets \ell_0$
	\label{lineFPAEllZero}
  \For{$\ell_1 = 1$ to $\ell_0$}
  \Comment{$\Oh{\ell_0 \timeSA}$ time}
  \label{lineFPAforloop}
    \State $\lambda_2 \gets \fnSelectLeaf(\ibeg{\intervalI}+\ell_1)$
	\label{lineFPALambdaTwo}
    \State $w_2 \gets \text{~lowest marked ancestor of~} \lambda_2$
		\label{lineFPAAncWTwo}
    \State $\ell_2 = \strdepth(\parent(w_2))+n_{w_2}+1$ 
	\label{lineFPAEllTwo}
    \If{$\ell_1 + \ell_2 > \ell_{\text{max}}$}
      \State $\ell_{\text{max}} \gets \ell_1 + \ell_2$
      \State $p \gets \ell_1$
    \EndIf
  \EndFor
  \State create reference $T[\ibeg{\intervalI}..\ibeg{\intervalI}+\ell_0]$ with reference ID stored in $w_1$ after having factorized up to position $\ibeg{\intervalI}+\ell_0$
  \label{lineFPAcreateRef}
  \Comment{store the reference at its \ST{} locus~$v$ and increment~$n_v$}
  \State \Output factor $T[\ibeg{\intervalI}..\ibeg{\intervalI}+p-1]$
  \State \Return by recursing on the interval $[\ibeg{\intervalI}+p..]$
  \end{algorithmic}
\end{algorithm}

For \iFPA{}, we mark the nodes dynamically.
By doing so, we can issue the following query for both flexible parsing variants.
For each suffix $T[i..]$, we query the lowest marked ancestor $w$ of the leaf with suffix number $i$. 
The string depth of $w$, obtained in $\Oh{\timeSA}$ time, is the length of the longest reference being a prefix of $T[i..]$.

If we allow \Oh{\timeSA \lg z} time per text position, we can omit the lowest marked ancestor data structure and 
make use of exponential search on \ST{}, as described in~\cite[Section~4]{koppl21nonoverlapping}.
Since $|R_y| \le y$, finding the reference~$R_y$ involves $\Oh{\lg y}$ node visits, and for each of them we pay $\Oh{\timeSA}$ time for computing its string length.
The total time is $\Oh{n \timeSA \lg z}$.
Since the LZ78-dictionary as well as the \iFPA{} dictionary are prefix-closed, the \ST{} superimposition by the LZ trie can be represented in \Oh{n} bits, and thus we need only $\Oh{n \lg \sigma}$ bits overall.

A pseudocode is given in \cref{algoFPA}.
We select the leaf with suffix number $\ibeg{\intervalI}$ (Line~\ref{lineFPALambdaOne})
and retrieve its lowest marked ancestor $w_1$ (Line~\ref{lineFPAAncWOne})
with a string depth of $\ell_0-1$ (Line~\ref{lineFPAEllZero}).
Subsequently, we check for each possible factor length $\ell_1 = |F_x| \in [1..\ell_0]$
via a loop (Line~\ref{lineFPAforloop}) whether a subsequent factor $F_{x+1}$ starting
at $\ibeg{\intervalI}+\ell_1$ gives the largest advance $|F_x| + |F_{x+1}|$ from $\ibeg{\intervalI}$ in the text. 
For that, 
select the leaf with suffix number $\ibeg{\intervalI}+|F_x|$ (Line~\ref{lineFPALambdaTwo})
and retrieve its lowest marked ancestor $w_2$ (Line~\ref{lineFPAAncWTwo})
with a string depth of $\ell_2-1$ (Line~\ref{lineFPAEllTwo}).
Let us focus on the loop of Line~\ref{lineFPAforloop} that takes \Oh{\ell_0 \timeSA} time per factor.
Naively, this time bound can get asymptotically larger than $\Oh{|\intervalI| \timeSA}$ if we only select short factors $\ell_{\max} \ll \ell_0$.
However, we can memoize the results for the lowest marked ancestor for the queried range $T[\ibeg{\intervalI}..\ibeg{\intervalI}+\ell_0]$
such that subsequent queries for the same positions can be performed in constant time until the respective text position has been factorized (the same argument as for the AC automaton approach in \cref{secACautomaton}).
By doing so, we process each position once.
Another matter that needs to be taken care of is the dictionary of references, which we update in Line~\ref{lineFPAcreateRef}.
The idea is to postpone this command for inserting the reference $T[\ibeg{\intervalI}..\ibeg{\intervalI}+\ell_0]$ until text position $\ibeg{\intervalI}+\ell_0$ has been factorized (otherwise we might create a factor using this reference that cannot be decompressed). 
To insert the reference $R = T[\ibeg{\intervalI}..\ibeg{\intervalI}+\ell_0]$, 
\begin{enumerate}
\item find its locus $w$ in \ST{}, which is the lowest marked ancestor of the leaf $\lambda_1$ (the leaf with suffix number~$\ibeg{\intervalI}$),
\item store in $w$ the index of $R$ and increment $n_w$ by one.
\end{enumerate}
The difference to \iLZEightFP{} is what references the parsing dictionary stores.
For \iLZEightFP{} it suffices to do the same steps but computing the LZ78 factors in the suffix tree simultaneously to \iLZEightFP{},
so only Line~\ref{lineFPAcreateRef} needs to be changed.

While the AC automaton solution for \iLZEightFP{} represents the parsing dictionary with a static data structure, we here need updates due to the use of the lowest marked ancestor data structure for finding a selectable reference. It is possible to make the dictionary for the suffix tree-based approach also static by first computing the LZ78 factorization but augment the edge witnesses with the first ending positions of the LZ78 factors they present. We then use these augmented values as keys for a weighted ancestor data structure (\cref{lemWeightedAncestor}) that allows us to jump from a leaf~$\lambda$ to its lowest edge witness ancestor that represents a reference ending before the suffix number of $\lambda$.
However, the time and space complexities do not change for this variation.

\begin{theorem}\label{thmSemiGreedy}
  We can compute \iLZEightFP{} or \iLZEightFPA{} of $T$ in 
  \Oh{n} time with $\Oh{n \lg n}$ bits of space,
  or in \Oh{n \timeSA \lg z} time with $\Oh{n \lg \sigma}$ bits of space,
  where $z$ is the number of factors of the LZ78 or \iFPA{} factorization for computing \iLZEightFP{} or \iLZEightFPA{}, respectively.
  (For \iLZEightFP{}, $z$ here includes the LZ78 factors while \iLZEightFPA{} creates as many references as factors.)
\end{theorem}

\begin{figure}[t]
  \begin{minipage}{0.48\linewidth}
    \centering LZ78~Factorization

  \includegraphics[width=\textwidth,page=1]{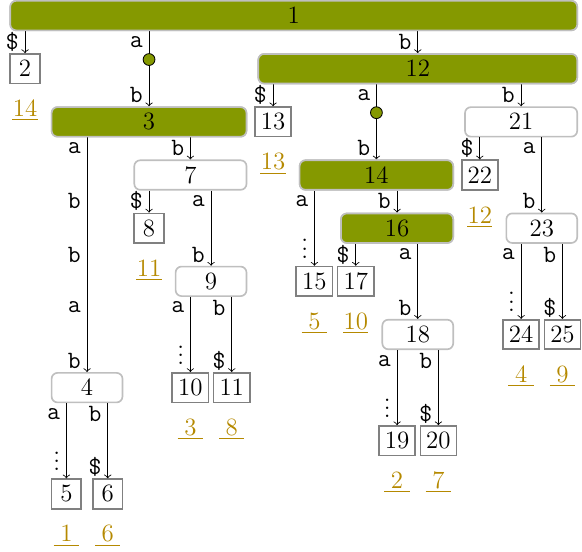}
  \end{minipage}
  \begin{minipage}{0.48\linewidth}
    \centering
    LZD~Factorization

  \includegraphics[width=\textwidth,page=2]{./suffixtree/superimposition}
  \end{minipage}
  \centering
  \caption{Suffix tree superimposition by the LZ trie of the LZ78 (left) and the LZD (right) factorization on our running example. 
    Small circles in green (\protect\PatternLegendS{fill=solarizedGreen}) on edges are implicit LZ trie nodes.
    The \ST{} nodes colored in green are full.
    In the LZD factorization, the edge witness with label 4 is colored in red (\protect\PatternLegendS{fill=solarizedRed}); 
    while not representing directly an LZ trie node, it witnesses factor $\texttt{ababb}$ represented by a green circle on its incoming edge.
  The LZD factorization is not prefix-closed, for instance the nodes marked in green are not connected in the suffix \emph{trie}.}
  \label{figSuffixTreeImposition}
\end{figure}

\section{LZ78 Factorization Variants}\label{lzEightVariations}
In what follows, we show applications of the technique of \cref{secFlexParseLinTime} for variants of the LZ78 factorization, 
and also give solutions to their substring compression variants.

\subsection{Lempel--Ziv Double (LZD)}
LZD~\cite{goto15lzd} is a variation of the LZ78 factorization.
A factorization $F_1 \cdots F_z$ of $T$ is LZD if
$F_x = R_1 \cdot R_2$ with
$R_1, R_2 \in \{F_1,\ldots,F_{x-1}\} \cup \Sigma$ such that
$R_1$ and $R_2$ are respectively the longest possible prefixes of $T[\dst_x..]$ and of $T[\dst_x+|R_1|..]$,
where $\dst_x$ denotes the starting position of $F_x$.
When computing $F_x$, the LZD dictionary stores the phrases
$\{F_y \cdot F_{y'} : y,y' \in [1..x-1]\} \cup \{F_y c : y \in [0..x-1], c \in \Sigma\}$.
See \cref{figFactLZD} for an example of the LZD factorization.
Like in LZ78, references are factors.

\begin{figure}
  \begin{minipage}{0.47\linewidth}
  \includegraphics[page=3,clip]{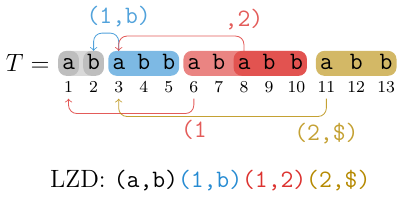}
  \end{minipage}
  \hfill
  \begin{minipage}{0.52\linewidth}
    \caption{LZD factorization of our running example\J{ $T = \texttt{ababbababbabb}$},
    given by
    $F_1 =\texttt{ab}$,
    $F_2 =\texttt{abb} = F_1 \cdot \texttt{b}$,
    $F_3 =\texttt{ababb} = F_1 \cdot F_2$, and
    $F_4 =\texttt{abb\$} = F_2 \cdot \$$.
    The factor $F_3$ has two referred indexes, which are visualized by two arrows pointing to the referred factors.
    We encode $F_4$ with the artificial character $\$$ to denote that we ran out of characters.
  }
  \label{figFactLZD}
  \end{minipage}
\end{figure}

A computational problem for LZD and LZMW is that their dictionaries are not prefix-closed in general. 
Consequently, their LZ tries are not necessarily connected subgraphs of the suffix trie: 
Indeed, some nodes in the suffix trie can be "jumped over" (cf.~\cref{figSuffixTreeImposition}).
Unfortunately, a requirement of the $\Oh{n}$-bits representation of the LZ trie built upon \ST{} is that the LZ trie is prefix-closed~\cite{fischer18lz}, which is not the case in the following two variations. In that light, we give up compact data structures and let each node store its exploration counter explicitly.
Another negative consequence is that the search for the longest reference cannot be performed by the exponential search used in \cref{secFlexParseLinTime} because some \ST{} node on a path from the \ST{} root to a locus of an LZD factor may not be edge witnesses.
It thus looks like that using the dynamic marked ancestor data structure of \cref{lemLowestMarkedAncestor} is the only efficient way for that task.

Finally, given we just have computed the factor $F_x = F_w \cdot F_y$,
we need to find the edge witness $v$ of $F_x$ and mark it.
For that, we search the locus of $F_x$, which is either represented by $v$, 
or is on an edge to $v$.
To find this locus, it suffices to traverse the path from the root to the leaf with suffix number $\dst_x$.
The number of visited nodes is upper bounded by the factor length $|F_x|$, and thus we traverse $\Oh{n}$ nodes in total.
For the substring compression problem, we want to get rid of the linear dependency in the text length.
For that, we jump to the locus with the weighted ancestor data structure introduced in \cref{lemWeightedAncestor}.
Our result is as follows.
\begin{theorem}\label{thmLZD}
There is a data structure that,
given a query interval $\intervalI$,
can compute LZD in $\Oh{\LZDSub}$ time, where $\LZDSub$ is the number of LZD factors.
This data structure can be constructed in \Oh{n} time, 
using $\Oh{n \lg n}$ bits of working space.
\end{theorem}
\begin{proof}
We list the following steps in \cref{algoLZD}.
For each factor~$F_x$ we want to determine, 
compute first the lowest marked ancestor~$w_1$ of the \ST{} leaf~$\lambda$ with suffix number $\dst_x$ (Line~\ref{lineLZDAncWOne}).
Say the exploration counter of $w_1$ plus the string depth of $w_1$'s parent is $\ell_1$ (Line~\ref{lineLZDEllOne}), 
compute the lowest marked ancestor~$w_2$ of the leaf with suffix number $\dst_x+\ell_1$ (Line~\ref{lineLZDAncWTwo}).
Given the exploration counter of $w_2$ combined with the string depth of $w_2$'s parent is $\ell_2$ (Line~\ref{lineLZDEllTwo}), 
$F_x$ has length $\ell_1+\ell_2$,
and its locus is on the path between $w_1$ and $\lambda$ on the string depth $\ell_1+\ell_2$ (Line~\ref{lineLZDcreateRef}).
After finding the locus~$u$ of $F_x$, we mark~$u$, 
increase its exploration counter, and continue with processing $F_{x+1}$.

In total, for each factor~$F_x$, 
we use 
\begin{itemize}
	\item two lowest marked ancestor queries, 
	\item a weighted ancestor query, and 
	\item mark a node or increase its exploration counter.
\end{itemize}
Each step takes constant time.
\end{proof}

\begin{algorithm}
  \caption{Computing LZD within the bounds claimed in \cref{thmLZD}.}
  \label{algoLZD}
  \begin{algorithmic}[1]
  \Require query interval $\intervalI$, root is marked
  \If{$\intervalI = \emptyset$}
  \Return
  \EndIf
  \State $\lambda_1 \gets \fnSelectLeaf(\ibeg{\intervalI})$
  \Comment{$\Oh{\timeSA}$ time}
  \State $w_1 \gets \text{~lowest marked ancestor of~} \lambda_1$
	\label{lineLZDAncWOne}
  \If{$w_1 = $ root}
  \Comment{left-hand side of factor is literal}
  \State $\ell_1 \gets 1$
  \State \Output $T[\ibeg{\intervalI}]$ as left-hand side of factor 
  \Else
	\State $\ell_1 \gets \strdepth(\parent(w_1))+n_{w_1}$
	\label{lineLZDEllOne}
  \State \Output reference stored it $w_1$ as left-hand side of factor 
  \EndIf
  \State $\lambda_2 \gets \fnSelectLeaf(\ibeg{\intervalI}+\ell_1)$
  \label{lineLZDsecondLeaf}
  \Comment{$\Oh{\timeSA}$ time}
  \State $w_2 \gets \text{~lowest marked ancestor of~} \lambda_2$
	\label{lineLZDAncWTwo}
  \If{$w_2 = $ root}
  \Comment{right-hand side of factor is literal}
  \State $\ell_2 \gets 1$
  \State \Output $T[\ibeg{\intervalI}+\ell_1]$ as right-hand side of factor
  \Else
	\State $\ell_2 \gets \strdepth(\parent(w_2))+n_{w_2}$
	\label{lineLZDEllTwo}
  \State \Output reference stored it $w_2$ as right-hand side of factor
  \EndIf

  \State create reference for factor $T[\ibeg{\intervalI}..\ibeg{\intervalI}+\ell_1+\ell_2-1]$
  \label{lineLZDcreateRef}
  \Comment{access the locus with weighted ancestor query}
  \State \Return by recursing on the interval $[\ibeg{\intervalI}+\ell_1+\ell_2..]$
  \end{algorithmic}
\end{algorithm}

\subsection{Lempel--Ziv--Miller--Wegman (LZMW)}
The factorization $T = F_1 \cdots F_z$ is the LZMW parsing of $T$ if,
for every $x \in [1..z]$, $F_x$ is the longest
prefix of $T[\dst_x..]$ with $F_x \in \{ F_{y-1} F_{y} : y \in [2..\dst_x - 1] \} \cup \Sigma$ where 
$\dst_x = 1 + \sum_{y=1}^{x-1} |F_y|$.
The LZMW dictionary for computing $F_x$ is the set of strings $\bigcup_{y \in [2..x-1]} ( F_{y-1} \cdot F_{y} ) \cup \Sigma$.
See \cref{figFactLZMW} for an example of the LZMW factorization.

\begin{figure}
  \begin{minipage}{0.47\linewidth}
  \includegraphics[page=2,clip]{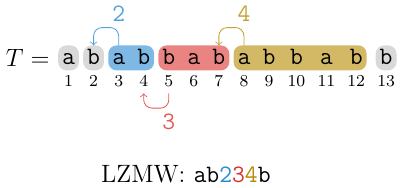}
  \end{minipage}
  \hfill
  \begin{minipage}{0.45\linewidth}
  \includegraphics[width=\textwidth,page=4]{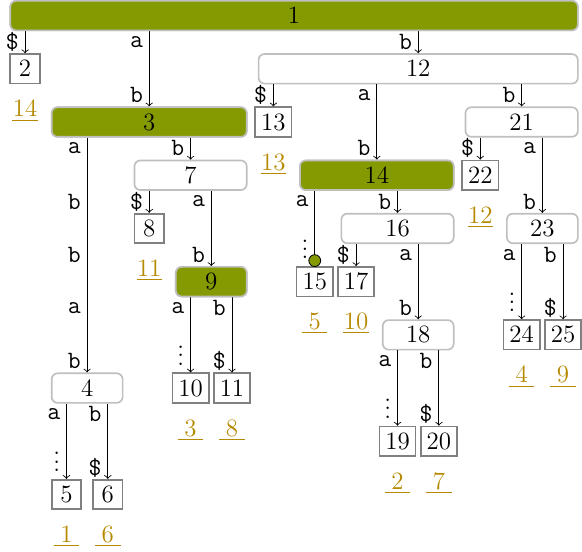}
  \end{minipage}

  \caption{LZMW factorization of our running example $T = \texttt{ababbababbabb}$,
    given by
    $F_1 =\texttt{a}$,
      $F_2 =\texttt{b}$,
      $F_3 =\texttt{ab} = F_1 F_2$,
      $F_4 =\texttt{bab} = F_2 F_3$,
      $F_5 =\texttt{abbab} = F_3 F_4$, and
      $F_6 =\texttt{b}$.
    \textbf{Left:} We encode the factors such that a number denotes the index of a referred factor.
      Writing $x$ means that we refer to the string $F_{x-1}F_{x}$.
			\textbf{Right:} \ST{} superimposed by the LZ trie for the LZMW factorization decorating $F_x \cdot F_{x+1}$, where $F_4 F_5$ is an implicit LZ trie node on the edge leading to the leaf with suffix number 5.
  }
  \label{figFactLZMW}
\end{figure}

\begin{theorem}\label{thmLZMW}
There is a data structure that,
given a query interval $\intervalI$,
can compute LZMW in $\Oh{\LZMWSub}$ time, where $\LZMWSub$ is the number of LZMW factors.
This data structure can be constructed in \Oh{n} time, 
using $\Oh{n \lg n}$ bits of working space.
\end{theorem}
\begin{proof}
For each factor~$F_x$, 
compute the lowest marked ancestor~$v$ of the \ST{} leaf with suffix number $\dst_x$.
The sum of the exploration counter of $v$ with the string depth of its parent is 
the length of $F_x$.
Finally, we mark the locus of $F_{x-1}F_x$ in \ST{}.
\J{}The difference to LZD is that
(a) we mark the locus of $F_{x-1} \cdot F_x$ instead of $F_x$ and that
(b) we have one lowest marked ancestor query instead of two per factor.
\end{proof}
To change \cref{algoLZD} to compute the LZMW factorization, we remove all commands concerning the leaf~$\lambda_2$
starting with Line~\ref{lineLZDsecondLeaf}, and store the reference $F_{x-1} \cdot F_x$ instead of just $F_x$ at Line~\ref{lineLZDcreateRef}.

\subsection{\texorpdfstring{\Oh{n \lg \sigma}}{O(n lg sigma)}-Bits Solution}

We can get down to $\Oh{n \lg \sigma}$ bits of space if we allow the time to linearly depend on $n$.
Then we can afford to mark all \ST{} ancestors of the loci of the computed LZD factors with an extra bit vector~$B$ of length $n$.
By doing so, the set of nodes marked by $B$ forms a sub-graph of the suffix tree,
whose lowest nodes leading to a suffix tree leaf we can query by exponential search using level ancestor queries.
We hence can find $w_1$ in Line~\ref{lineLZDAncWOne} of \cref{algoLZD} in $\Oh{\lg z}$ time.
On the downside, we mark \Oh{n} nodes in $B$ during the computation of LZD\@.

Next, we need to get rid of 
(1) the lowest marked ancestor data structure and 
(2) the $\Oh{n \lg n}$-bits representation of the exploration counters.
First, we overcome the computation bottleneck to find the locus $u$ of the new factor we want to insert into the LZ trie.
This locus $u$ is on the path on the first selected leaf $\lambda_1$ and its lowest marked ancestor~$w_1$. 
We can find $u$ with a binary search on the depth $d \mapsto \levelanc(\lambda_1, d)$ with $\strdepth$ being the key to evaluate.
Finding $u$ takes $\Oh{\timeSA \lg n}$ additional time per factor.

Second, we store the exploration values of each edge witness in a dynamic balanced binary search tree using $\Oh{z \lg n} = \Oh{n \lg \sigma}$ bits of space, where $z = \Oh{n / \log_\sigma n}$ denotes the number of factors of LZD or LZMW.
Lookups and updates of any exploration value thus costs $\Oh{\lg z} = \Oh{\lg n}$ time.

\begin{theorem}\label{thmLZDSpace}
We can compute LZD or LZMW in $\Oh{n \timeSA \lg n} = \Oh{n \lg^{1+\epsilon} n}$ time,
using $\Oh{n \lg \sigma}$ bits of working space.
\end{theorem}

\section{Evaluation}\label{secEvaluation}

For the following evaluations, we used a machine with Intel Xeon Gold 6330 CPU, Ubuntu 22.04, and gcc version 11.4.
The evaluations were evaluated on datasets from the Pizza\&Chili~\cite{ferragina08compressed} corpus, the Canterbury corpus~\cite{arnold97corpus}, and the Calgary corpus~\cite{bell89modeling}.

\subsection{Flexible Parsing}\label{secEvalFlexParse}
As far as we are aware of, only compiled binaries for computing a binary-encoded \iLZWFP{} or \iLZWFP{} factorization on Solaris and Irix platforms are available~\footnote{\url{https://www.dcs.warwick.ac.uk/~nasir/work/fp/}, accessed 12th of September 2024.}.
We here provide source code for both flexible parsings written in Python at \url{https://github.com/koeppl/lz78flex}.
The code is not algorithmically engineered, but suitable as a reference implementation.
We conducted experiments to empirically evaluate the change in the number of factors compared with (classic) LZ78.
In \cref{tabFlex} we observe that there is no instance where a flexible parsing instance produces more factors than LZ78.
Instead, the factorizations have up to 87.13\% fewer factors that LZ78.
Despite that we have no guarantee about an upper bound on the factors for \iLZEightFPA{} in contrast to \iLZEightFP{}, \iLZEightFPA{} almost always has fewer factors than \iLZEightFP{}; an exception is the dataset \textsc{E.coli}.

\newcommand{\zCl}{\ensuremath{z_{\textup{78}}}}
\newcommand{\zFP}{\ensuremath{z_{\textup{\iLZEightFP}}}}
\newcommand{\zFPA}{\ensuremath{z_{\textup{\iLZEightFPA}}}}

\begin{table}[t]
\centerline{\begin{tabular}{l*{6}{r}}
\toprule
text & $n$ [K] & $\zCl$ [K] & $\zFP$ [K] & $\frac{\zFP}{\zCl}$ & $\zFPA$ [K] & $\frac{\zFPA}{\zCl}$
\\\midrule
\Dataset{E.coli} & \num{4638.69} & \num{491.11} & \num{488.36} & 99.44\% & \num{490.19} & 99.81\% \\
\Dataset{alice29.txt} & \num{148.48} & \num{28.73} & \num{27.87} & 97.03\% & \num{27.50} & 95.72\% \\
\Dataset{asyoulik.txt} & \num{125.18} & \num{25.59} & \num{24.82} & 97.00\% & \num{24.50} & 95.73\% \\
\Dataset{bib} & \num{111.26} & \num{21.46} & \num{20.40} & 95.05\% & \num{19.49} & 90.80\% \\
\Dataset{bible.txt} & \num{4047.39} & \num{490.81} & \num{472.51} & 96.27\% & \num{453.89} & 92.48\% \\
\Dataset{book1} & \num{768.77} & \num{131.07} & \num{128.07} & 97.71\% & \num{126.94} & 96.85\% \\
\Dataset{book2} & \num{610.86} & \num{102.51} & \num{98.76} & 96.34\% & \num{96.11} & 93.75\% \\
\Dataset{fields.c} & \num{11.15} & \num{2.79} & \num{2.66} & 95.40\% & \num{2.58} & 92.46\% \\
\Dataset{grammar.lsp} & \num{3.72} & \num{1.07} & \num{1.03} & 95.89\% & \num{0.98} & 91.13\% \\
\Dataset{lcet10.txt} & \num{419.24} & \num{71.12} & \num{68.78} & 96.71\% & \num{67.37} & 94.72\% \\
\Dataset{paper1} & \num{53.16} & \num{12.17} & \num{11.74} & 96.52\% & \num{11.49} & 94.44\% \\
\Dataset{paper2} & \num{82.20} & \num{17.34} & \num{16.81} & 96.97\% & \num{16.60} & 95.75\% \\
\Dataset{paper3} & \num{46.53} & \num{10.91} & \num{10.60} & 97.21\% & \num{10.49} & 96.16\% \\
\Dataset{paper4} & \num{13.29} & \num{3.65} & \num{3.53} & 96.74\% & \num{3.51} & 96.30\% \\
\Dataset{paper5} & \num{11.95} & \num{3.41} & \num{3.30} & 96.74\% & \num{3.29} & 96.36\% \\
\Dataset{paper6} & \num{38.11} & \num{9.15} & \num{8.82} & 96.41\% & \num{8.66} & 94.68\% \\
\Dataset{plrabn12.txt} & \num{471.16} & \num{84.11} & \num{82.25} & 97.80\% & \num{81.54} & 96.95\% \\
\Dataset{progc} & \num{39.61} & \num{9.46} & \num{9.09} & 96.13\% & \num{8.87} & 93.72\% \\
\Dataset{progl} & \num{71.65} & \num{13.62} & \num{12.95} & 95.05\% & \num{12.43} & 91.24\% \\
\Dataset{progp} & \num{49.38} & \num{9.81} & \num{9.32} & 94.94\% & \num{8.99} & 91.58\% \\
\Dataset{world192.txt} & \num{2408.28} & \num{309.45} & \num{290.48} & 93.87\% & \num{269.63} & 87.13\% \\
\Dataset{xargs.1} & \num{4.23} & \num{1.34} & \num{1.30} & 97.02\% & \num{1.28} & 95.46\% \\
 \\
			\bottomrule
		\end{tabular}
	}\caption{Flexible parsing variants of LZ78 compared with the LZ78 factorization on the number of factors of the given datasets.
		The text length and the number of factors are given in thousands (divided by $10^{3}$, marked by [K]).
	}
\label{tabFlex}
\end{table}

\subsection{Substring Compression}\label{secEvalEight}

\begin{table}
\centerline{\begin{tabular}{l*{6}{r}}
\toprule
text & $n$ [M] & $z_{78}$ [M] & \iTernary{} & \iCics{}
\\\midrule
\Dataset{E.coli} & \num{4.64} & \num{0.49} & \num{0.92} & \num{28.50} \\
\Dataset{bible.txt} & \num{4.05} & \num{0.49} & \num{0.95} & \num{21.81} \\
\Dataset{dblp.xml.00001.1} & \num{104.86} & \num{3.94} & \num{1.40} & \num{35.23} \\
\Dataset{dblp.xml.00001.2} & \num{104.86} & \num{3.97} & \num{1.41} & \num{37.08} \\
\Dataset{dblp.xml.0001.1} & \num{104.86} & \num{3.95} & \num{1.43} & \num{40.22} \\
\Dataset{dblp.xml.0001.2} & \num{104.86} & \num{4.25} & \num{1.36} & \num{41.18} \\
\Dataset{fib41} & \num{267.91} & \num{0.42} & \num{13.56} & \num{289.62} \\
\Dataset{fib46} & \num{1836.31} & \num{1.52} & \num{33.73} & \num{1143.04} \\
\Dataset{rs.13} & \num{216.75} & \num{0.44} & \num{10.60} & \num{223.19} \\
\Dataset{tm29} & \num{268.44} & \num{0.62} & \num{9.61} & \num{230.44} \\
\Dataset{world192.txt} & \num{2.47} & \num{0.31} & \num{0.92} & \num{15.20} \\
 \\
			\bottomrule
		\end{tabular}
	}\caption{LZ78 factorization speed benchmark.
	  File sizes in the second column are given in megabytes. 
		The third column $z_{78}$ denotes the number of computed factors divided by $10^6$.
	  The two last columns measure the average delay in microseconds per factor for each respective factorization algorithm ($\mu s / z$).
	}
\label{tabBenchmark}
\end{table}

Curious about the practicality of our proposed solution with the suffix tree, 
we prepared a preliminary implementation based on the LZ78 factorization algorithm with suffix trees~\cite{fischer18lz}.
We call this implementation \iCics{} for \emph{computation in compressed space}.
For comparison, we took the LZ78 implementation in tudocomp~\cite{dinklage17tudocomp} using a ternary trie~\cite{bentley97multikeyquicksort} as the LZ78 trie implementation.
We call this implementation \iTernary{} in the following. See~\cite{fischer17lz78} for implementation details of \iTernary{}.
Both \iCics{} and \iTernary{} are implemented in C++ in the tudocomp framework available in the \texttt{cics} branch at \url{https://github.com/tudocomp/tudocomp}.
As \ST{} implementation, we use the class \texttt{cst\_sada} of the SDSL library~\cite{gog14sdsl},
a C++ implementation of the compressed suffix tree of \cite{sadakane07compressed}.

In what follows, we benchmark the time needed for the substring compression problem of LZ78 when the entire text should be compressed.
In this setting, \iCics{} first builds \ST{} at the precomputation time, while starting the factorization at query time.
We hence omit the construction of \ST{} in the time benchmark.
Nevertheless, we observe in \cref{tabBenchmark} that \iTernary{} processes a factor faster on average on all datasets.

Reasons for the bad performance of \iCics{} is the slow performance of queries such as \parent{}, \strdepth{}, and \fnSelectLeaf{}.
All these methods require in the \texttt{cst\_sada} implementations unpredictable long jumps in memory because the underlying implementation heavily depends on rank/select support data structures.
Unfortunately, memoizing the results of these calls in dynamic lookup tables did not give expected speedups.
Despite the bad performance, we are positive that \ST{} implementations specifically addressing the experienced bottlenecks in answering these queries will improve the performance of \ST{}-based factorization algorithms.

\section{Outlook}\label{secOutlook}
An open problem is whether we can compute LZD or LZMW in \Oh{z} space within \Oh{n} time,
which would be optimal in the general case.
We further wonder about the minimum required space for representing the $\LPF$ tables of all suffixes of $T$.

\myblock{LZD Derivates}
The idea of LZD spawns new research directions.
\begin{itemize}
	\item
	The first question would be whether it makes sense to generalize LZD to refer to at most $k$ factors instead of at most two. By doing so, the string lengths in the parse dictionary can grow much faster.
	It seems that the computational time complexity then also linearly depends on $k$.
	\item
	A problem for the optimality of LZD is that it is not prefix-free,
	which is a property that can be exploited to find strings for which LZD gives a grammar that is
	larger than the size of the smallest grammar by a factor of $\Om{n^3}$~\cite{badkobeh17two}.
	Also note that the flexible parsing applied to LZ78 or LZW makes both parsings optimal with respect to the minimum number of factors, 
	which however does not apply to LZD or LZMW, since both parsings are not prefix-closed in general.
	\item
	It is possible to extend LZD to a collage system~\cite{kida03collage},
	where we are also allowed to select prefixes of any element in the dictionary as a factor.
	Such an extension, paired with the flexible parsing, may give better lower bounds on 
	the sizes (now with respect to the smallest collage system), and might be also worth to try in practice.
\end{itemize}

\myblock{Other Factorizations}
It is interesting to study substring compression queries or internal queries for other factorizations such as
block palindrome factorizations~\cite{goto18block} or repetition factorizations~\cite{inoue16computing}.

\paragraph{Acknowledgements}
This research was supported by JSPS KAKENHI with grant number \texttt{JP23H04378} and the Research Grant for Young Scholars Funded by Yamanashi Prefecture with grant number \texttt{2291}.

\bibliographystyle{abbrvnat}

\clearpage
\appendix

\end{document}